\documentclass[11pt]{article}
\usepackage[round]{natbib}

\usepackage{fullpage}

\usepackage{amsmath}
\usepackage{graphicx}
\usepackage{amssymb}
\usepackage{pdfpages}
\usepackage{amsthm}
\usepackage{microtype}
\usepackage{subfig}
\usepackage{float}
\usepackage{algorithm}
\usepackage{nicefrac}
\usepackage[noend]{algpseudocode}
\usepackage{xcolor}
\usepackage[colorlinks=true, allcolors=blue]{hyperref}
\usepackage{cleveref}

\theoremstyle{plain}
\newtheorem{theorem}{Theorem}[section]
\theoremstyle{definition}
\newtheorem{definition}[theorem]{Definition}

\newtheorem{lemma}[theorem]{Lemma}
\newtheorem{fact}[theorem]{Fact}
\newtheorem{corollary}[theorem]{Corollary}

\newtheorem*{remark}{Remark}

\newcommand{\R}{\mathbb{R}}

\newcommand{\calM}{\mathcal{M}}
\newcommand{\E}{\mathop{\mathbb{E}}}
\DeclareMathOperator*{\argmax}{arg\,max}

\title{Differentially Private Learning of Exponential Distributions: \\ Simple Algorithms and Tight Bounds}

\author{
Bar Mahpud\\
Bar-Ilan University\\
\texttt{mahpudb@biu.ac.il}
\and
Or Sheffet\\
Bar-Ilan University\\
\texttt{or.sheffet@biu.ac.il}
}

\date{}

\begin{document}

%

%
\maketitle

\begin{abstract}
We study the problem of learning exponential distributions under differential privacy. 
Given $n$ i.i.d.\ samples from $\mathrm{Exp}(\lambda)$, the goal is to privately estimate $\lambda$ so that the learned distribution is close in total variation distance to the truth.
We present a simple pure $\epsilon$-differentially private algorithm that avoids the classical dependence on the true value of $\lambda$. 
Our method leverages a structural property of the exponential distribution: its $(1-1/e)$-quantile equals $1/\lambda$, allowing us to estimate the rate parameter directly via private quantile estimation. 
The resulting learner is both conceptually simple and sample-efficient, achieving near-optimal guarantees.
We further extend the method to Pareto distributions via a logarithmic reduction, prove nearly matching lower bounds using group privacy arguments, and show how approximate $(\epsilon,\delta)$-DP removes the need for externally supplied parameter bounds.
Together, these results give the first tight characterization of exponential distribution learning under differential privacy using a simple $\lambda$-free approach.
\end{abstract}

\section{Introduction}
\label{sec:intro}

The exponential distribution occupies a central role in probability and statistics. 
It arises naturally in modeling interarrival times of Poisson processes \citep{Ross1996,Asmussen2003}, 
waiting times in queueing systems \citep{Kleinrock1975,Daskin2021}, 
lifetimes in reliability analysis and survival studies \citep{Smith2012,Johnston2012}, 
and interarrival processes in communication and computer networks \citep{Kleinrock1975,BertsekasGallager}. 
Its simplicity---being defined by a single rate parameter $\lambda$---makes it a canonical distribution both for theory and for practice. 
Moreover, the exponential law serves as a fundamental building block for more complex families, including the Pareto and Weibull distributions, which are widely used in economics, finance \citep{Mandelbrot1960,Gabaix2009}, and network modeling \citep{Crovella1997,Feldmann1998}. 
Accurate estimation of the rate parameter $\lambda$ is therefore a fundamental statistical task.

In modern applications, however, the data from which such parameters must be estimated are often sensitive. 
Waiting times may reveal information about individual users, interarrival processes may correspond to private communication events, and survival data may capture confidential medical information. 
As a result, there is a pressing need for algorithms that allow reliable learning of the exponential distribution while protecting the privacy of individuals. 
\emph{Differential privacy} (DP) \citep{Dwork2006} has emerged as the gold standard for such guarantees, providing provable bounds on the information leakage about any single data point. 
Yet, despite significant progress on private distribution learning in general, the case of exponential distributions has not been studied in detail.

Our work lies at the intersection of private distribution learning and heavy-tailed statistics.
Classical results on distribution learning under differential privacy include private estimation of means, variances, histograms, and quantiles (see, e.g., \cite{Dwork2006,Smith2011,Bun2015,Bun2018}). 
More recent works have studied private parameter estimation for specific distribution families, including Gaussian mean and covariance estimation \citep{Kamath2019,Kamath2021}, multinomial distributions \citep{Kifer2012,Acharya2019}, and heavy-tailed settings \citep{Kamath2021,Ashtiani2020}. 
In particular, \cite{Karwa2017FiniteSD} introduced finite-sample lower bounds for private distribution learning via group privacy arguments, which we adapt in establishing our lower bounds.
Complementary information-theoretic perspectives have also been developed. 
In the local differential privacy model, \cite{BarnesChenOzgur2020} analyze how privatization affects Fisher information and derive unified lower bounds for several estimation problems, clarifying how estimation complexity scales with the privacy parameter. 
More recently, \cite{McMillanSmithUllman2022} study instance-optimal estimation under differential privacy and construct locally minimax estimators for one-parameter exponential families. Their results are applicable to the exponential distribution, yet obtain $\lambda$-dependent bounds that deteriorate as $\lambda \to 0$.
And yet, the exponential and Pareto distributions~---despite their central role in modeling waiting times and heavy-tailed phenomena~--- have not previously received tight, distribution-specific treatment under differential privacy. 
To the best of our knowledge, no prior work provides differentially private algorithms with near-tight sample complexity guarantees for learning exponential or Pareto distributions in total variation distance. 
Our results bridge this gap, providing near-optimal private learners in this fundamental setting.

In this work our goal is to design algorithms that, given $n$ i.i.d.\ samples from $\mathrm{Exp}(\lambda)$, produce an estimate $\tilde{\lambda}$ such that the learned distribution $\mathrm{Exp}(\tilde{\lambda})$ is $\alpha$-close in total variation distance to the true distribution in total variation distance while adhering to $\epsilon$-differential privacy. 
As we exhibit in Section~\ref{subsec:exponential}, this requirement reduces to obtaining a multiplicative $(1\pm \alpha)$-approximation to $\lambda$.
The central question we address is therefore: how many samples are required to achieve this goal under differential privacy?

Towards that end, we develop in Section~\ref{sec:private-alg} a simple algorithmic strategy based on a structural property unique to the exponential distribution: its $(1-1/e)$-quantile equals $1/\lambda$. By privately estimating this quantile, we directly obtain an estimate of $\lambda$ without relying on its value. The resulting algorithm is both simple and robust, and achieves near-optimal sample complexity under pure differential privacy.

Our $\epsilon$-DP algorithm requires an apriori knowledge of lower- and upper-bounds, $\lambda_{\min}$ and $\lambda_{\max}$ resp., on the value of $\lambda$.
In section~\ref{sec:lower-bound} we establish a nearly matching lower bound on our sample complexity showing that the  doubly logarithmic $\log\log(\lambda_{\max}/\lambda_{\min})$ dependency is necessary,
using group privacy arguments from~\cite{Karwa2017FiniteSD}. Our lower-bound shows that our sample complexity are essentially tight, up to logarithmic factors. 

Finally, in Section~\ref{sec:approx-dp} we address the challenge of initialization. Bypassing the lower-bound, we give an approximate $(\epsilon,\delta)$-DP algorithm that retrieves such bounds on $\lambda$ using the histogram-based algorithm of~\cite{Vadhan2017}, thus removing the dependency on these apriori bounds on $\lambda$.

In summary, this work makes the following contributions:
\begin{itemize}
    \item We introduce the first simple pure differentially private algorithm for learning exponential distributions via direct quantile estimation, achieving $(1\pm \alpha)$ multiplicative accuracy of the rate parameter.
    \item We extend our techniques to the Pareto distribution, demonstrating how exponential learning enables private estimation for heavy-tailed families.
    \item We establish lower bounds showing that our sample complexity guarantees are tight up to logarithmic factors and require prior bounds under pure DP.
    \item We show how approximate DP enables removal of externally supplied parameter bounds through private initialization.
\end{itemize}

Taken together, these results highlight the exponential distribution as a rich and instructive case study in private statistics. 
Our findings not only resolve its sample complexity up to logarithmic terms but also provide new algorithmic tools and structural insights that can and should extend to broader classes of distributions.

\section{Preliminaries}
\label{sec:preliminaries}

\subsection{Differential Privacy}
\label{subsec:dp}

\begin{definition}[Differential Privacy]
A randomized algorithm $\calM : \mathcal{X}^n \to \mathcal{Y}$ is 
\emph{$(\epsilon,\delta)$-differentially private} (approximate differentially private) if for all neighboring datasets 
$D,D'\in\mathcal{X}^n$ and all measurable $S\subseteq\mathcal{Y}$,
\[
\Pr[\calM(D)\in S] \;\le\; e^\epsilon \Pr[\calM(D')\in S] + \delta
\]
when $\delta=0$ we say that the algorithm is \emph{$\epsilon$-differentially private} (pure differentially private).
\end{definition}

\begin{definition}[Laplace Mechanism]
Let $f:\mathcal{X}^n \to \R^d$ be a function with $\ell_1$-sensitivity
\[
\Delta f \;=\; \max_{D,D'} \|f(D)-f(D')\|_1
\]
over neighboring datasets $D,D'$.  
The \emph{Laplace mechanism} releases
\[
\calM(D) \;=\; f(D) + (Z_1,\dots,Z_d), \qquad Z_i \overset{\text{i.i.d.}}{\sim}\mathrm{Lap}(\Delta f/\epsilon)
\]
This mechanism is $\epsilon$-differentially private.
\end{definition}

\begin{lemma}[Basic Composition]
If $\calM_1$ is $(\epsilon_1,\delta_1)$-DP and $\calM_2$ is $(\epsilon_2,\delta_2)$-DP, then their sequential composition $(\calM_1,\calM_2)$ is $(\epsilon_1+\epsilon_2,\ \delta_1+\delta_2)$-DP.  
In particular, composing $k$ pure $\epsilon$-DP mechanisms yields $k\epsilon$-DP.
\end{lemma}

\subsection{Probabilistic Background}\label{subsec:prob-background}

\begin{definition}[Total Variation Distance]
For two probability distributions $P$ and $Q$ over a measurable space $(\mathcal{X},\mathcal{F})$, 
the \emph{total variation distance} is defined as
\[
\mathrm{TV}(P,Q) \;=\; \sup_{A \in \mathcal{F}} |P(A) - Q(A)|
= \tfrac{1}{2} \int_{\mathcal{X}} |p(x)-q(x)| \, dx
\]
\end{definition}

\begin{lemma}[Dvoretzky–Kiefer–Wolfowitz (DKW) Inequality]
Let $X_1,\dots,X_n \sim P$ i.i.d.\ with CDF $F$, and let $F_n$ denote the empirical CDF. Then for any $\eta>0$,
\[
\Pr\!\left[\sup_x |F_n(x)-F(x)| > \eta\right] \;\le\; 2e^{-2n\eta^2}
\]
\end{lemma}

\begin{lemma}[Multiplicative Chernoff Bound]
Let $X_1,\dots,X_n$ be independent $\{0,1\}$-valued random variables with mean $\mu=\E[\sum_i X_i]$. Then for any $\alpha\in(0,1)$,
\[
\Pr\!\left[\sum_i X_i \ge (1+\alpha)\mu\right] \le \exp(-\tfrac{\alpha^2}{3}\mu)\]
\[
\Pr\!\left[\sum_i X_i \le (1-\alpha)\mu\right] \le \exp(-\tfrac{\alpha^2}{2}\mu)
\]
\end{lemma}

\subsection{Exponential Distribution}\label{subsec:exponential}

For $\lambda>0$, the exponential distribution $\mathrm{Exp}(\lambda)$ has density
$f(x) = \lambda e^{-\lambda x}$ for $x\ge 0$. Its mean is $1/\lambda$, and its CDF is $F(x)=1-e^{-\lambda x}$.

We also recall below a fact and a corollary relating the exponential family and total variation, which will be used throughout; the proofs of both statements are deferred to the appendix for completeness (See~\Cref{apx-fact:tv-exp} and~\Cref{apx-cor:tv-exp}).

\begin{fact}\label{fact:tv-exp}
    Let $\mathrm{Exp}(\lambda_1), \mathrm{Exp}(\lambda_2)$ be two exponential distributions and assume $\lambda_1 \geq \lambda_2$. Denote $a = \frac{\ln(\lambda_1) - \ln(\lambda_2)}{\lambda_1 - \lambda_2}$, then:
    \[
    TV(\mathrm{Exp}(\lambda_1), \mathrm{Exp}(\lambda_2)) = e^{-\lambda_2 \cdot a} - e^{-\lambda_1 \cdot a}
    \]
\end{fact}

\begin{corollary}\label{cor:tv-exp}
    If the rate parameters satisfy $\lambda_1 \in (1 \pm \alpha)\lambda_2$, then
    \[
    TV\!\left(\mathrm{Exp}(\lambda_1), \mathrm{Exp}(\lambda_2)\right) \leq \alpha
    \]
\end{corollary}

The canonical way for estimating the parameter $\lambda$ is using its MLE estimator.

\begin{definition}[Maximum Likelihood Estimator for the Exponential Distribution]
Let $X_1,\dots,X_n \sim \mathrm{Exp}(\lambda)$ be i.i.d.\ samples.  
The \emph{maximum likelihood estimator (MLE)} for the rate parameter $\lambda$ is
\[
\hat{\lambda} \;=\; \frac{1}{\bar{X}}
\;=\; \left(\frac{1}{n}\sum_{i=1}^n X_i\right)^{-1}
\]
\end{definition}
In~\citet{McMillanSmithUllman2022}, the main algorithm builds on a private version of the MLE estimator. In Appendix~\ref{apx-sec:private-alg} we present this privatize MLE algorithm as well (tailored to the exponential distribution), which works by estimating privately the mean of the samples. Note however that DP-mean estimation requires that we add noise proportional to the range of the input points~--- which all lie inside an interval of length $\frac{\log(n)}{\lambda}$ when the data is drawn from $\mathrm{Exp}(\lambda)$. This results in sample complexity of $\tilde O(\frac {1}{\epsilon\alpha\lambda})$, which suggests that when $\lambda$ is small, the private MLE algorithm is costly. In contrast, our algorithm avoids this dependency on $\nicefrac 1\lambda$. In addition, in Appendix~\ref{apx-sec:private-alg} we give an algorithm that enjoys ``the best of both worlds'': First it (privately) checks whether $\lambda\geq 1$ or not. If so, it applies the private MLE estimator, and if not it applies our private quantile-based estimation. This allows us to obtain a sample complexity that benefit from large values of $\lambda$ while being independent of $\nicefrac 1 \lambda$ when $\lambda$ is small. 

\subsection{Pareto Distribution}\label{subsec:pareto-prelim}

The \emph{Pareto distribution} is a classical heavy-tailed law that is parameterized by a minimum (scale) parameter $x_m > 0$ and a shape parameter $\alpha_p > 0$, and is supported on $[x_m, \infty)$. 
Its probability density function is
\[
    f(x) \;=\; \alpha_p \cdot \frac{x_m^{\alpha_p}}{x^{\alpha_p+1}}
    \qquad\text{for } x \geq x_m.
\]
We denote this law by $\mathrm{Pareto}(x_m, \alpha_p)$.

\begin{remark}
    If $X \sim \mathrm{Pareto}(x_m, \alpha_p)$, then 
    \[
        Y \;=\; \ln\!\left(\tfrac{X}{x_m}\right) \;\sim\; \mathrm{Exp}(\alpha_p).
    \]
    Thus, estimation of $\alpha_p$ can be reduced to exponential learning. 
\end{remark}

\section{Pure Differentially Private Algorithm for Exponential Distribution Learning}
\label{sec:private-alg}
We now turn to our main technical contribution: learning the parameter $\lambda$ of an exponential distribution under pure $\epsilon$-differential privacy. 
Our target is to output an estimate $\tilde{\lambda}$ such that the learned distribution $\mathrm{Exp}(\tilde{\lambda})$ is close in total variation distance to the true distribution $\mathrm{Exp}(\lambda)$. 
By Corollary~\ref{cor:tv-exp}, it suffices to achieve a $(1\pm \alpha)$–multiplicative approximation of $\lambda$, since this immediately implies $\mathrm{TV}(\mathrm{Exp}(\tilde{\lambda}),\mathrm{Exp}(\lambda)) \leq \alpha$. 

Our approach relies on a structural property of the exponential distribution: its $(1-1/e)$-quantile equals $1/\lambda$. 
By privately estimating this quantile, we directly obtain an estimate of $\lambda$ without relying on its value. 
This leads to a simple pure DP learner whose analysis we develop in this section.

\subsection{Learning $\lambda$ from the $(1-1/e)$-Quantile}\label{subsec:exp-quantile-only}
We now present a simplified approach for privately estimating $\lambda$ that bypasses MLE entirely. The method relies on the fact that   the $(1 - 1/e)$-quantile of $\mathrm{Exp}(\lambda)$ equals $Q_{1 - 1/e} = \frac{1}{\lambda}$.
Thus, a multiplicative $(1 \pm \alpha)$-approximation to $Q_{1 - 1/e}$ directly yields a multiplicative $(1 \pm \alpha)$-approximation to $1/\lambda$, and therefore to $\lambda$ itself.
We use a binary search on powers of $(\frac 1 {1-\nicefrac \alpha 2})$ that lie inside the interval $[\lambda_{\min}, \lambda_{\max}]$ for known bounds $0 < \lambda_{\min} < \lambda_{\max}$, with the comparison at each step implemented via a private counting query. The output of the search is $\tilde{q}$, which we interpret as $\tilde{q} \approx \tfrac 1\lambda$, so our final estimator is $\tilde{\lambda} = \tfrac 1 {\tilde{q}}$.

\begin{algorithm}[ht]
\caption{Private Exponential Learning via $1/e$-Quantile}
\label{alg:exp-quantile}
\begin{algorithmic}
\State \textbf{Input:} $P = \{x_i\}_{i=1}^n \sim \mathrm{Exp}(\lambda)$, bounds $0 < \lambda_{\min} < \lambda_{\max}$, target error $\alpha \in (0,1)$, privacy budget $\epsilon > 0$.
\State \textbf{Output:} Parameter estimator $\tilde{\lambda}$.
\Procedure{Quantile-Learning}{$P, \lambda_{\min}, \lambda_{\max}, \alpha, \epsilon$}
\State Denote $T = \left\lceil \log_2\!\left(\frac{2\log(\lambda_{\max}/\lambda_{\min})}{\alpha}\right) \right\rceil$
\State Denote $p_0 = \lambda_{\min}, p_1 = \lambda_{\min}(\frac{1}{1-\nicefrac{\alpha}{2}}), p_2 = \lambda_{\min}(\frac{1}{1-\nicefrac{\alpha}{2}})^2,  p_3 = \lambda_{\min}(\frac{1}{1-\nicefrac{\alpha}{2}})^3,\dots, p_{2^T} = \lambda_{\max}$
\State Denote $L \gets 0$, $U \gets 2^T$
\For{$t = 1$ \textbf{to} $T$}
    \State $m \gets \frac{U+L}{2}$
    \State $\tilde{q} \gets \frac{1}{n}|\{x_i \in P : x_i < p_m\}| + \mathrm{Lap}\left(\frac{T}{\epsilon n}\right)$
    \If{$\tilde{q} \ge 1 - \nicefrac{1}{e} + \nicefrac{\alpha}{2e}$}
        \State $U \gets m$
    \ElsIf{$\tilde{q} \le 1 - \nicefrac{1}{e} - \nicefrac{\alpha}{2e}$}
        \State $L \gets m$
    \Else
        \State \Return $\tilde{\lambda} = 1/\tilde{q}$
    \EndIf
\EndFor
\State \Return $\perp$
\EndProcedure
\end{algorithmic}
\end{algorithm}

\begin{lemma}[Privacy of~\Cref{alg:exp-quantile}]\label{lem:exp-quantile-privacy}
    Algorithm~\ref{alg:exp-quantile} is $\epsilon$-differentially private.
\end{lemma}
The proof is deferred to \Cref{apx-lem:exp-quantile-privacy}.

\begin{lemma}[Utility of~\Cref{alg:exp-quantile}]\label{lem:exp-quantile-utility}
    Fix target accuracy $\alpha \in (0,1)$ and failure probability $\beta \in (0,1)$. Let $T = \left\lceil \log_2\!\left(\frac{2\log(\lambda_{\max}/\lambda_{\min})}{\alpha}\right) \right\rceil$. If
    \[
        n \ \ge\ \max\left\{
            \frac{2eT}{\epsilon \alpha} \ln\left( \frac{2T}{\beta} \right),
            \frac{2}{\alpha^2} \ln\left( \frac{2T}{\beta} \right)
        \right\}
    \]
    then with probability at least $1 - \beta$ Algorithm~\ref{alg:exp-quantile} outputs $\tilde{\lambda}$ satisfying
    $
        \tilde{\lambda} \in \left[ (1-\alpha)\lambda,\ (1+\alpha)\lambda \right].
    $
\end{lemma}
\begin{proof}[Proof sketch]
The algorithm performs a binary search over $2^T=2\log(\lambda_{\max}/\lambda_{\min})/\alpha$ checkpoints, where each step compares the (noisy) empirical CDF at a candidate $p$ to $1-1/e$.

\emph{Sampling error.} By DKW, $\Pr[\sup_x|F_n(x)-F(x)|>\alpha/2]\le 2e^{-2n(\alpha/2)^2}$. A union bound over the $T$ comparisons shows that taking
$
n\ \ge\ \frac{2}{\alpha^2}\ln\!\Big(\frac{2T}{\beta}\Big)
$
ensures all empirical CDF values deviate from the truth by at most $\alpha/2$ with probability $\ge 1-\beta/2$.

\emph{Noise error.} Each comparison releases one count with Laplace noise of scale $b=T/(\epsilon n)$. Requiring that every noisy perturbation be at most $\alpha/(2e)$ (so that the CDF-side error budget totals $\alpha/2$ from privacy) holds with probability $\ge 1-\beta/2$ if $n \;\ge\; \frac{2eT}{\epsilon \alpha}\,
   \ln\!\Big(\frac{2T}{\beta}\Big).$


\emph{From CDF to parameter accuracy.} For $X \sim \mathrm{Exp}(\lambda)$, the $(1-1/e)$-quantile equals $1/\lambda$. 
An additive CDF error of at most $\alpha/(2e)$ around the level $1-1/e$ corresponds to at least $\alpha/(2e)$ probability mass in the exponential distribution. 
This amount of probability mass separates the true quantile $1/\lambda$ from the neighboring points $1/((1\pm\alpha)\lambda)$. 
Therefore, controlling the CDF error at this level guarantees that the estimated quantile lies in $\big[\tfrac{1}{(1+\alpha)\lambda},\,\tfrac{1}{(1-\alpha)\lambda}\big]$, and thus its reciprocal $\tilde\lambda$ falls in $[(1-\alpha)\lambda,(1+\alpha)\lambda]$.

Combining all requirements via the union bound yields the stated sample-size condition with success probability $1-\beta$. Full details are deferred to \Cref{apx-lem:exp-quantile-utility}.
\end{proof}

The following theorem summarizes the properties of Algorithm~\ref{alg:exp-quantile}.

\begin{theorem}\label{thm:exp-quantile-main}
    Algorithm~\ref{alg:exp-quantile} is $\epsilon$-differentially private and, for any $\alpha \in (0,1)$ and $\beta \in (0,1)$, achieves an $(1\pm \alpha)$-approximation to $\lambda$ with probability at least $1-\beta$, provided
    \begin{align*}
    n \;\ge\; O\!\Bigg(
       \max\Bigg\{&
          \frac{\ln(\tfrac{\log(\lambda_{\max}/\lambda_{\min})}{\alpha})}
               {\epsilon \alpha}\,
          \ln\!\Big(\tfrac{\log(\nicefrac{\lambda_{\max}}{\lambda_{\min}})/\alpha}{\beta}\Big),
          \frac{1}{\alpha^2}\,
          \ln\!\Big(\tfrac{\log(\nicefrac{\lambda_{\max}}{\lambda_{\min}})/\alpha}{\beta}\Big)
       \Bigg\}
    \Bigg)
    \end{align*}
\end{theorem}

\begin{proof}
    Privacy follows from \Cref{lem:exp-quantile-privacy}. Utility follows from \Cref{lem:exp-quantile-utility}. 
\end{proof}

\section{Applications to Pareto Distributions}
\label{sec:pareto}



The exponential distribution is not only fundamental in its own right but also serves as a key building block for other families. 
In particular, the \emph{Pareto distribution}---one of the most widely studied heavy-tailed laws in economics, finance, and network modeling---has probability density
\[
f(x) = \alpha_p \cdot \frac{x_m^{\alpha_p}}{x^{\alpha_p+1}}, \qquad x \ge x_m>0,
\]
where $x_m$ is the scale parameter (minimum value) and $\alpha_p$ is the shape parameter.  
Through the logarithmic transform $y = \ln(x/x_m)$, a Pareto$(x_m,\alpha_p)$ random variable reduces exactly to $\mathrm{Exp}(\alpha_p)$.  
This reduction allows us to transfer our private exponential learning algorithm directly to the Pareto setting.

A subtlety arises, however, in the case when $x_m$ is unknown: direct release of the sample minimum is infeasible under differential privacy due to its high sensitivity.  
Instead, one must combine quantile-based estimation with the logarithmic reduction in order to privately and accurately recover both $(\alpha_p, x_m)$.  

We develop and analyze these reductions formally, showing how our exponential learner extends to Pareto distributions while preserving privacy and accuracy.  
The full algorithms, proofs, and utility guarantees are deferred to \Cref{apx-sec:pareto}.

\section{Lower Bound for Exponential Distribution Learning under Pure Differential Privacy}
\label{sec:lower-bound}

To complement our algorithmic results, we establish lower bounds on the sample complexity of privately learning exponential distributions under pure DP. 
Our construction follows a standard packing argument: we identify a family of rate parameters $\Lambda$ that are well-separated in TV-distance, and then invoke DP generalization bounds to show that distinguishing them requires a large number of samples. 
Our lower bound is based on the group privacy approach of~\cite{Karwa2017FiniteSD}, which provides a powerful technique for  deriving lower bounds in private distribution learning. 
This demonstrates that our upper bounds from Section~\ref{sec:private-alg} are essentially tight, up to logarithmic factors.

Denote $\Lambda$ as the set \[
\Lambda = \left\{\lambda_{\min}, \lambda_{\min}(1+8\alpha), \lambda_{\min}(1+8\alpha)^2, \dots, \lambda_{\max}\right\}
\]
and consider the following family of $t = |\Lambda|$ distributions:
\[
\{P_i = \mathrm{Exp}(\lambda_i):~ \lambda_i \in \Lambda\}
\]
so the rate parameters range geometrically from $\lambda_{\min}$ up to $\lambda_{\max}$ with factor $(1+8\alpha)$.

First, we show that $\forall i \neq j:~ TV(P_i, P_j) > \alpha$. Fix $i > j$ so that $\lambda_i > \lambda_j$. Then, by~\Cref{fact:tv-exp}:
\[
    TV(\mathrm{Exp}(\lambda_i), \mathrm{Exp}(\lambda_j)) = e^{-\lambda_j \cdot a} - e^{-\lambda_i \cdot a}
\]
where
\[
a = \frac{\ln(\lambda_i) - \ln(\lambda_j)}{\lambda_i - \lambda_j} = \frac{\ln(\frac{\lambda_i}{\lambda_j})}{\lambda_i - \lambda_j}
\]
Then we get:
\begin{equation*}
\begin{split}
    & TV(\mathrm{Exp}(\lambda_i), \mathrm{Exp}(\lambda_j)) = e^{-\lambda_i \cdot \frac{\ln((\lambda_i/\lambda_j)}{\lambda_i - \lambda_j}} - e^{-\lambda_j \cdot \frac{\ln(\lambda_i/\lambda_j)}{\lambda_i - \lambda_j}} \\
    & = e^{-\frac{\ln(\lambda_i/\lambda_j)}{(\lambda_i/\lambda_j) - 1}} - e^{-(\lambda_i/\lambda_j)\frac{\ln(\lambda_i/\lambda_j)}{(\lambda_i/\lambda_j) - 1}} \overset{r = \nicefrac{\lambda_i}{\lambda_j}}{=} r^{-\frac{1}{r - 1}} - r^{-\frac{r}{r - 1}} \\
    & = r^{-\frac{1}{r - 1}} \cdot (1 - r^{-1}) 
\end{split}
\end{equation*}

Now, the following lemma proves that $TV(P_i, P_j) > \alpha$.
\begin{lemma}\label{lem:tv-alpha-exp}
Let $T(r)$ denote the total variation distance between two exponential distributions whose
rates differ by a ratio $r\ge 1$:
\begin{align*}
T(r) &:= \operatorname{TV}\!\big(\mathrm{Exp}(r\lambda),\mathrm{Exp}(\lambda)\big) \notag = r^{-\tfrac{1}{r-1}} - r^{-\tfrac{r}{r-1}} \notag \\
     &= r^{-\tfrac{1}{r-1}}\!\left(1 - \tfrac{1}{r}\right)
\end{align*}
Fix any $\alpha\in(0,\tfrac12)$. Then there exists an absolute constant $c$ such that,
with $r=1+c\alpha$, one has $T(r)\ge \alpha$. In particular,
the choice $c=8$ suffices.
\end{lemma}

The proof is deferred to~\Cref{apx-lem:tv-alpha-exp}. Now we can state the following lemma.


\begin{lemma}[Formal lower bound]
\label{lem:lower-bound-formal}
Fix parameters $0<\alpha<\tfrac12$, $0<\beta<\tfrac12$, and an interval
$[\lambda_{\min},\lambda_{\max}]$ with $0<\lambda_{\min}<\lambda_{\max}$.
Let $\calM$ be any $\epsilon$-differentially private algorithm that, given
$n$ i.i.d.\ samples from $\mathrm{Exp}(\lambda)$, outputs $\tilde\lambda$ such that
$\Pr[\,\tilde\lambda\in[(1-\alpha)\lambda,(1+\alpha)\lambda]\,]\ge 1-\beta$
for every $\lambda\in[\lambda_{\min},\lambda_{\max}]$.
Then there exists a universal constant $c>0$ such that the sample size must satisfy
\[
n\ \ge\ \frac{c}{\epsilon\,\alpha}\;
\ln\!\Bigg(\frac{\ \frac{\log(\lambda_{\max}/\lambda_{\min})}{\alpha}\ }{\beta}\Bigg)
\]
\end{lemma}

\begin{corollary}
    Any $\epsilon$-DP learner achieving TV-error at most $\alpha$ (via
    Corollary~\ref{cor:tv-exp}) also requires
    $n=\Omega\!\big(\tfrac{1}{\epsilon\alpha}\ln(\tfrac{\ln(\lambda_{\max}/\lambda_{\min})/\alpha}{\beta})\big)$.
\end{corollary}

\begin{proof}[Proof of~\Cref{lem:lower-bound-formal}]
    Fix $i$ and let $\tilde \lambda$ denote the output of any $\epsilon$-DP exponential distribution learner. Then, by construction, for all $j \ne i$,
    \[
    \Pr_{\substack{P \sim P_i^n}}\!\left[\tilde \lambda \in (1 \pm \alpha)\tfrac{1}{(1+8\alpha)^j}\right] \;\leq\; \beta
    \]
    In contrast, we can use the following lemma.
    \begin{lemma}[Lemma 6.1 of~\cite{Karwa2017FiniteSD}]
        For every pair of distributions $\mathbb{D}_{\theta_0}, \mathbb{D}_{\theta_1}$, every $(\epsilon, \delta)$-differentially private algorithm $M(x_1, \dots, x_n)$, if $M_{\theta_0}, M_{\theta_1}$ are two induced marginal distributions on the output of $M$ evaluated on input dataset $X_1, \dots, X_n$ sampled i.i.d. from $\mathbb{D}_{\theta_0}$ and $\mathbb{D}_{\theta_1}$ respectively, $\epsilon' = 6\epsilon n \|\mathbb{D}_{\theta_0} - \mathbb{D}_{\theta_1}\|_{tv}$ and $\delta' = 4e^{\epsilon'}n\delta\|\mathbb{D}_{\theta_0} - \mathbb{D}_{\theta_1}\|_{tv}$, then, for every event $E$,
        \[
        M_{\theta_0}(E) \leq e^{\epsilon'}M_{\theta_1}(E) + \delta'
        \]
    \end{lemma}
    The lemma above implies
    \begin{equation*}
    \begin{split}
    & \Pr_{\substack{P \sim P_i^n}}\!\left[\tilde \lambda \notin (1 \pm \alpha)\tfrac{1}{(1+8\alpha)^i}\right]
    \\
    & \quad \geq \sum_{\substack{j=1 \\ j \neq i}}^t \Pr_{P \sim P_i^n}\!\left[\tilde \lambda \in (1 \pm \alpha)\tfrac{1}{(1+8\alpha)^j}\right] \\
    & \quad \geq 
    \sum_{\substack{j=1 \\ j \neq i}}^t e^{6n\epsilon \,\mathrm{TV}(P_i, P_j)}
    \Pr_{P \sim P_j^n}\!\left[\tilde \lambda \in (1 \pm \alpha)\tfrac{1}{(1+8\alpha)^j}\right] \\
    & \quad \geq\ t \cdot e^{6n\epsilon \alpha} (1-\beta)
    \;\;\overset{\beta \leq 1/2}{\geq}\; \tfrac{1}{2}t \cdot e^{6n\epsilon \alpha}
    \end{split}
    \end{equation*}
    Noting that $\lambda_{\min}(1+8\alpha)^t = \lambda_{\max}$ we get that
    \begin{equation*} 
    t = \frac{\log(\lambda_{\max}/\lambda_{\min})}{\log(1+8\alpha)}
    \ge \frac{1}{8\alpha}\log(\lambda_{\max}/\lambda_{\min})
    \end{equation*}
    we conclude that
    $
    \frac{1}{16\alpha}\ln(\nicefrac{\lambda_{\max}}{\lambda_{\min}}) \cdot e^{6n\epsilon \alpha} \leq \beta
    $,
    which rearranges to the following lower bound on the required sample size:
    \[
    n \geq \frac{1}{6\epsilon\alpha}\,\ln\!\left(\frac{\frac{\ln(\nicefrac{\lambda_{\max}}{\lambda_{\min}})}{16\alpha}}{\beta}\right)\qedhere
    \]
\end{proof}

\section{Finding Initial Bounds Using Approximate Differential Privacy}
\label{sec:approx-dp}

The algorithm of Section~\ref{sec:private-alg} required prior bounds $(\lambda_{\min},\lambda_{\max})$ on the range of plausible $\lambda$ in order to calibrate privacy noise. 
As shown by our lower bound in Section~\ref{sec:lower-bound}, such knowledge of $(\lambda_{\min},\lambda_{\max})$ is in fact \emph{necessary} for pure differential privacy, since otherwise no algorithm can achieve nontrivial accuracy. 
In this section we show that under approximate $(\epsilon,\delta)$-DP we can eliminate this requirement: using a median-based private routine, we can obtain valid initial bounds directly from the data, and then feed them into our algorithm. 
Thus the extra logarithmic dependence on $\lambda_{\max}/\lambda_{\min}$ is an artifact of pure DP, not of the problem itself.

To eliminate the need for externally provided bounds $(\lambda_{\min},\lambda_{\max})$ under $(\epsilon,\delta)$-DP, 
we design a private subroutine that learns such bounds directly from the data. 
The key idea is to approximate the median of the exponential distribution using noisy counts over dyadic intervals; 
since the median is tightly concentrated around $\ln(2)/\lambda$, this gives a natural anchor point. 
By expanding outward from this approximate median, we obtain lower and upper bounds $\tilde \lambda_{\min}$ and $\tilde \lambda_{\max}$ that safely contain the true median and hence allow us to calibrate the later learning algorithms. 
This approach is based on the differentially private histogram technique of~\cite{Vadhan2017}, which provides an efficient way to identify approximate median under privacy constraints. 
The procedure is summarized in~\Cref{alg:median-dp} below.

\begin{algorithm}[ht]
\caption{Finding initial private approximated bounds $0 < \tilde \lambda_{\min} < \tilde \lambda_{\max}$}
\label{alg:median-dp}
\begin{algorithmic}
\State \hspace*{-\algorithmicindent} \textbf{Input:} $P = \{x_i\}_{i=1}^n \sim Exp(\lambda)$, privacy budget $\epsilon, \delta$.
\State \hspace*{-\algorithmicindent} \textbf{Output:} bounds $0 < \tilde \lambda_{\min} < \tilde \lambda_{\max}$.
\Procedure{Finding-Bounds}{$P, \epsilon, \delta$}
    \State For all $k \in \mathbb{Z}$, set $c_k(P) = \frac{1}{n}|\{x_i \mid x_i \in (2^k, 2^{k+1})\}|$.
    \State For all $c_k(P)$, set $\hat c_k(P) = 0$ if $c_k(P) = 0$. Else set $\hat c_k(P) = c_k(P) + Z_k$, where each $Z_k \gets \text{Lap}(\nicefrac{2}{\epsilon n})$ independently.
    \State Set $S = \{k \mid \hat c_k(P) \geq \frac{2}{\epsilon n}\log(\nicefrac{2}{\delta}) + \frac{1}{n}\}$.
    \If{$|S| \geq 1$}
        \State Let $k^* \gets \argmax\limits_{k \in S} \hat c_k(P)$
        \State \Return $\tilde \lambda_{\min} = \ln(2)(2^{k^* + 1})^{-1}, \quad \tilde \lambda_{\max} = \ln(2)(2^{k^* - 1})^{-1}$
    \Else
        \State \Return $\perp$
    \EndIf
\EndProcedure
\end{algorithmic}
\end{algorithm}

\begin{lemma}\label{lem:bounds}
    \Cref{alg:median-dp} is $(\epsilon, \delta)$-differentially private. Moreover, if
    \[
    n \geq \max\left\{\frac{800}{\epsilon}\ln(\frac{2}{\delta\beta}), 5000\ln(\frac{2}{\beta})\right\}
    \]
    then
    \[
    \Pr[\tilde \lambda_{\min} < \lambda < \tilde \lambda_{\max}] \geq 1 - \beta
    \]
\end{lemma}

    The proof is deferred to \Cref{apx-lem:bounds}.
    


We now state the main guarantee of Algorithm~\ref{alg:median-dp}.  


\begin{theorem}\label{thm:bounds-dp}
    There exists an $(\epsilon, \delta)$-differentially private algorithm that, for any $\alpha \in (0,1)$ and $\beta \in (0,1)$, achieves a $(1\pm \alpha)$-approximation to $\lambda$ with probability at least $1-\beta$, provided
    \[
    n \geq O\!\left(
      \max \left\{
        \begin{aligned}
          &
            \tfrac{\ln(\nicefrac{1}{\alpha})}{\epsilon \alpha}
            \ln\!\left( \tfrac{\ln(\nicefrac{1}{\alpha})}{\beta} \right)
          ,
          &\tfrac{1}{\alpha^2} \ln\!\left(\tfrac{\ln(\nicefrac{1}{\alpha})}{\beta}\right)
        \end{aligned}
      \right\}
    \right)
    \]
\end{theorem}

\section{Conclusion}
\label{sec:conclusion}

We introduced a simple differentially private algorithm for learning exponential distributions based on direct private quantile estimation. This approach avoids dependence on fragile mean-based estimators and leads to clean guarantees with near-optimal sample complexity. We further showed how exponential learning enables private estimation for Pareto distributions, and proved matching lower bounds up to logarithmic factors. Finally, we demonstrated that approximate DP removes the need for externally supplied parameter bounds.

Looking ahead, several directions remain open. Can our techniques be extended to multi-parameter exponential families or to mixtures? What are the precise tradeoffs for other heavy-tailed distributions? And more broadly, can our strategy be systematically developed for other problems in private statistics? We hope this work motivates further exploration of learning additional heavy-tailed distributions under differential privacy.

\bibliographystyle{plainnat} 
\bibliography{bibliography}

\clearpage
\appendix
\thispagestyle{empty}





\section{Additional Algorithms}
\label{apx-sec:private-alg}
This section presents additional learning procedures for exponential distributions that complement the main quantile-based algorithm described in Section~\ref{sec:private-alg}. 

These methods are not required for the main results of the paper but are included for completeness. In particular, we describe an MLE-based private learner as well as an adaptive strategy combining multiple estimators.
\subsection{Learning $\lambda$ from the MLE}\label{apx-subsec:mle}
We first describe an alternative approach based on the maximum likelihood estimator (MLE). 
Recall that for exponential data, the MLE of the rate parameter equals the reciprocal of the sample mean. 
Under differential privacy, however, the sample mean must be carefully privatized, which requires controlling sensitivity via clipping and private range estimation.

We therefore construct a pipeline consisting of:
(i) private quantile estimation to obtain a clipping range, followed by
(ii) private release of the clipped mean.

\begin{algorithm}[H]
\caption{Unbounded Quantile Estimation}
\label{alg:quantile}
\begin{algorithmic}
\State \hspace*{-\algorithmicindent} \textbf{Input:} $P = \{x_i\}_{i=1}^n \sim Exp(\lambda)$, bounds $0 < \lambda_{\min} < \lambda_{\max}$, quantile parameter $\theta \in (0,1)$, privacy budget $\epsilon>0$.
\Procedure{Quantile-Estimation}{$P, \lambda_{\min}, \lambda_{\max}, \theta, \epsilon$}
    \State $\tilde T \gets (1-\theta) + Lap(\frac{2}{\epsilon n})$
    \For{$i = 0, 1, 2 \dots$}
        \If{$\frac{1}{n}|\{x_i \in P \mid~x_i < 2^i\cdot\lambda_{\min}\}| + Lap(\frac{4}{\epsilon n}) \geq \tilde T$}
            \State \Return $2^i\cdot\lambda_{\min}$
        \EndIf
    \EndFor
    \Return $\perp$
\EndProcedure
\end{algorithmic}
\end{algorithm}

\begin{lemma}[Privacy of~\Cref{alg:quantile}]\label{lem:quantile-privacy}
    \Cref{alg:quantile} is $\epsilon$-differentially private.
\end{lemma}
    The proof is deferred to \Cref{apx-lem:quantile-privacy}.

\begin{lemma}[Utility of~\Cref{alg:quantile}]\label{lem:quantile-utility}
    \Cref{alg:quantile} returns a $6$-approximation of the $(1-\theta)$-quantile of $P$ w.p. $\geq 1-\beta$, given that $\nicefrac{1}{10} \leq \theta \leq \nicefrac{9}{10}$ and:
    \[
    n \geq \max\left\{\frac{5}{\epsilon} \ln\left(\frac{4\log\left(\nicefrac{\lambda_{\max}}{\lambda_{\min}}\right)}{\beta}\right), 200\ln\left(\nicefrac{4}{\beta}\right)\right\}
    \]
\end{lemma}
\begin{proof}[Proof sketch]
The analysis combines two ingredients:  
(i)~the Dvoretzky–Kiefer–Wolfowitz inequality, which bounds the deviation of the empirical CDF from the true CDF; and  
(ii)~concentration of Laplace noise, which controls the perturbations added to the CDF evaluations and the threshold.  
By ensuring both sampling error and noise error are at most $1/20$ (with probability at least $1-\beta$), we show that the returned quantile is within a constant factor of the true $(1-\theta)$-quantile.  
A detailed proof is provided in \Cref{apx-lem:quantile-utility}.
\end{proof}

\begin{corollary}\label{cor:range}
Let $\tilde Q_{1-\theta}$ be the output of \Cref{alg:quantile}. Fix a target failure probability $2\beta\in(0,1)$. There exists an absolute constant $c_0\ge 1$ from \Cref{lem:quantile-utility} such that, if we set
\[
\tilde R \;=\; C\,\tilde Q_{1-\theta}\,\ln n
\quad\text{with}\quad
C \;\ge\; \frac{c_0}{\ln(1/\theta)}\Big(1+\frac{\ln(1/\beta)}{\ln n}\Big)
\]
then with probability at least $1-2\beta$ all $n$ samples lie in $[0,\tilde R]$. In particular,
\[
\tilde R \;=\; O\!\left(\frac{1}{\lambda}\ln\!\Big(\nicefrac{1}{\theta}\Big)\,\ln n\right)
\]
\end{corollary}

The proof is deferred to \Cref{apx-cor:range}.


Having obtained $\tilde R$, a private estimate of the range of the datapoints, we next turn to estimating $\lambda$ via the maximum likelihood approach. Since the exponential MLE is simply the inverse of the sample mean, our strategy is to clip the data to the estimated range and then release a noisy mean under the Laplace mechanism. This yields  \Cref{alg:private-mle}.

\begin{algorithm}[H]
\caption{Private Exponential Distribution MLE}
\label{alg:private-mle}
\begin{algorithmic}
\State \hspace*{-\algorithmicindent} \textbf{Input:} $P = \{x_i\}_{i=1}^n \sim Exp(\lambda)$, range $R$, privacy budget $\epsilon>0$.
\Procedure{Private-MLE}{$P, R, \epsilon$}
    \State $\forall i:~\bar{x}_i \gets \min\{x_i, R\}$ \Comment{Clipping}
    \State $\tilde s \gets \frac{1}{n}\sum\limits_{i=1}^n \bar{x}_i + Lap(\frac{R}{\epsilon n})$
    \State $\tilde \lambda \gets \frac{1}{\tilde s}$
    \State \Return $\tilde \lambda$
\EndProcedure
\end{algorithmic}
\end{algorithm}

\begin{lemma}[Privacy of~\Cref{alg:private-mle}]\label{lem:mle-privacy}
    \Cref{alg:private-mle} is $\epsilon$-differentially private.
\end{lemma}
The proof is deferred to \Cref{apx-lem:mle-privacy}.

\begin{lemma}[Utility of~\Cref{alg:private-mle}]\label{lem:mle-utility}
    \Cref{alg:private-mle} returns an $(1 \pm \alpha)$-approximation of $\lambda$ w.p. $\geq 1-\beta$, given that:
    \[
    n \geq O\left(\max\left\{\frac{R\ln(\frac{1}{\beta})}{\epsilon(\alpha - e^{-\lambda R})}, \frac{\ln(\frac{1}{\beta})}{\alpha^2}\right\}\right)
    \]
\end{lemma}
\begin{proof}[Proof sketch]
The estimator $\tilde{\lambda}$ is defined as the reciprocal of the privatized clipped mean. 
Its error arises from three sources: (i) sampling error of the empirical mean, controlled by a multiplicative Chernoff bound for exponentials; (ii) clipping bias, which is at most $e^{-\lambda R}/\lambda$ since the exponential tail mass beyond $R$ is $\exp(-\lambda R)$; and (iii) Laplace noise added for privacy, which with probability $1-\frac \beta 2$ is bounded by $\tfrac{R}{\epsilon n}\ln(\tfrac 2 \beta)$. 
Combining these bounds shows that $\tilde{s}$, the noisy clipped mean, is within an additive $\alpha/\lambda$ of the true mean provided
\[
n \;\ge\; O\!\left(\max\!\Big\{\tfrac{R\ln(1/\beta)}{\epsilon(\alpha - e^{-\lambda R})}, \;\tfrac{\ln(1/\beta)}{\alpha^2}\Big\}\right)
\]
under the assumption\footnote{$\alpha = O(1)$ and $e^{-\lambda R}=O(1/poly(n))$} $\alpha > e^{-\lambda R}$.
Taking reciprocals then yields that $\tilde{\lambda}$ is an $(1\pm \alpha)$-approximation to $\lambda$ with probability at least $1-\beta$. 
Full details are deferred to \Cref{apx-lem:mle-utility}.
\end{proof}

We now combine the two ingredients: the quantile-based range estimator and the private MLE. The resulting end-to-end algorithm takes raw samples, first derives a safe clipping range, and then computes a private estimate of $\lambda$ using the clipped noisy mean. The complete pipeline is summarized in \Cref{alg:exp-mle} below.

\begin{algorithm}[H]
\caption{Private Exponential Learning via MLE}
\label{alg:exp-mle}
\begin{algorithmic}
\State \hspace*{-\algorithmicindent} \textbf{Input:} $P = \{x_i\}_{i=1}^n \sim Exp(\lambda)$, bounds $0 < \lambda_{\min} < \lambda_{\max}$, target error $\alpha\in(0,1)$, privacy budget $\epsilon>0$.
\State \hspace*{-\algorithmicindent} \textbf{Output:} Parameter estimator $\tilde \lambda$.
\Procedure{MLE-Learning}{$P, \lambda_{\min}, \lambda_{\max}, \alpha, \epsilon$}
    \State $\tilde R \gets \Call{Quantile-Estimation}{P, \lambda_{\min}, \lambda_{\max}, \theta = \nicefrac{1}{10}, \nicefrac{\epsilon}{2}}$ \Comment{Range estimation}
    \State $\tilde \lambda \gets \Call{Private-MLE}{P, \tilde R \log (n), \nicefrac{\epsilon}{2}}$
    \State \Return $\tilde \lambda$
\EndProcedure
\end{algorithmic}
\end{algorithm}


\begin{theorem}\label{thm:main-theorem}
    \Cref{alg:exp-mle} is $\epsilon$-differentially private, and w.p. $\geq 1-\beta$ it returns a $(1 \pm \alpha)$-approximation of $\lambda$, given that:
    \begin{equation*}
    \begin{split}
    n \;\geq\; O\Bigg(\max\Bigg\{
        \frac{\ln(\tfrac{1}{\alpha})\ln(\tfrac{1}{\beta})\ln(n)}{\lambda\epsilon\alpha},
        \frac{\ln(\tfrac{1}{\beta})}{\alpha^2},\;
        \frac{\ln(\tfrac{\ln(\nicefrac{\lambda_{\max}}{\lambda_{\min}})}{\beta})}{\epsilon}\Big)
    \Bigg\}\Bigg)
    \end{split}
    \end{equation*}
\end{theorem}

\begin{proof}
    \emph{Privacy.} The algorithm is a sequential composition of the quantile/range subroutine run with privacy budget $\epsilon/2$ (\Cref{lem:quantile-privacy}) and the private MLE subroutine run with privacy budget $\epsilon/2$ (\Cref{lem:mle-privacy}). By basic composition, the overall privacy loss is at most $\epsilon$.

    \emph{Utility.} With the choice $\theta=1/10$ in the range routine and by its utility guarantee (\Cref{lem:quantile-utility}), with probability $\ge 1-\beta/2$ we obtain $\tilde Q_{1-\theta}$ within a constant factor of $Q_{1-\theta}=\frac{1}{\lambda}\ln(1/\theta)$, and hence (by \Cref{cor:range}) a clipping level $\Theta\!\Big(\frac{1}{\lambda}\cdot \ln n\Big)$.
    Consequently $e^{-\lambda \tilde R} = O(1/poly(n)) \le \alpha/2$. Plugging this $\tilde R$ into the MLE analysis (\Cref{lem:mle-utility} with budget $\epsilon/2$ and confidence $\beta/2$) yields the utility of~\Cref{alg:exp-mle}. Union-bounding the two stages gives success probability $\ge 1-\beta$. Collecting the dominant terms proves the stated big-$O$ bound.
\end{proof}

\subsection{Best of Both Algorithms}
\label{apx-subsec:best-of-both}
We now describe an adaptive strategy combining the MLE-based learner and the quantile-based learner. 
Since the two approaches have complementary sample-complexity behavior depending on the magnitude of $\lambda$, we can first obtain a coarse private estimate and then select the more efficient estimator accordingly.

\begin{algorithm}[ht]
\caption{Adaptive Private Exponential Learning (Best-of-Both)}
\label{alg:best-of-both}
\begin{algorithmic}
\State \textbf{Input:} $P=\{x_i\}_{i=1}^n$, bounds $0<\lambda_{\min}<\lambda_{\max}$, target error $\alpha\in(0,1)$, privacy budget $\epsilon>0$.
\State \textbf{Output:} $\tilde\lambda$.
\Procedure{Best-of-Both}{$P, \lambda_{\min}, \lambda_{\max}, \alpha, \epsilon$}
\smallskip
\State $\tilde\lambda_0 \gets \Call{Quantile-Learning}{P,\ \lambda_{\min},\ \lambda_{\max},\alpha=\nicefrac{1}{2},\ \nicefrac{\epsilon}{3}}$ \hfill (Alg.~\ref{alg:exp-quantile})
\smallskip
\If{$\tilde\lambda_0 \ge 2$} \Comment{confidently ``large''}
    \State $\tilde\lambda \gets \Call{MLE-Learning}{P,\ \lambda_{\min},\ \lambda_{\max},\ \alpha,\ \nicefrac{2\epsilon}{3}}$ \hfill (Alg.~\ref{alg:exp-mle})
\Else
    \State $\tilde\lambda \gets \Call{Quantile-Learning}{P,\ \lambda_{\min},\ \lambda_{\max},\ \alpha,\ \nicefrac{2\epsilon}{3}}$ \hfill (Alg.~\ref{alg:exp-quantile})
\EndIf
\State \Return $\tilde\lambda$
\EndProcedure
\end{algorithmic}
\end{algorithm}

Algorithm~\ref{alg:exp-quantile} can produce a constant-factor estimate
$\tilde\lambda_0$ of $\lambda$ using $O(1)$ accuracy. With such a coarse estimate, the events \(\tilde\lambda_0\ge 2\) and \(\tilde\lambda_0\le \tfrac12\) imply, respectively, that \(\lambda>1\) and \(\lambda<1\). The middle region \(\tilde\lambda_0\in(\tfrac12,2)\) corresponds to \(\lambda\in(\tfrac12,2)\), where both routes have linear terms within a factor at most~2, and either choice is acceptable up to constants; we default to the quantile route.

\begin{lemma}[Privacy of Alg.~\ref{alg:best-of-both}]
\label{lem:best-privacy}
Algorithm~\ref{alg:best-of-both} is $\epsilon$-differentially private.
\end{lemma}
The proof is deferred to \Cref{apx-lem:best-privacy}.


\begin{theorem}
\label{thm:best-main}
Fix $\alpha\in(0,1)$, $\beta\in(0,1)$ and $\epsilon>0$.
Algorithm~\ref{alg:best-of-both} is $\epsilon$-differentially private and,
with probability at least $1-\beta$, outputs $\tilde\lambda$ satisfying
\(\tilde\lambda\in[(1-\alpha)\lambda,(1+\alpha)\lambda]\),
provided that
\[
n \geq \tilde O\left(\max\left\{\min\left\{\frac{1}{\epsilon\alpha}, \frac{1}{\lambda\epsilon\alpha}\right\}, \frac{1}{\alpha^2}, \frac{1}{\epsilon}\right\}\right)
\]
Equivalently, up to polylogarithmic factors, the adaptive procedure matches the
best of the two options:
\begin{align*}
\text{if }\lambda \ge 1 \quad &\Rightarrow\quad 
  \tilde{O}\!\left(
    \max\left\{
      \tfrac{1}{\lambda\epsilon\alpha},\;
      \tfrac{1}{\epsilon},\;
      \tfrac{1}{\alpha^2}
    \right\}
  \right), \\[6pt]
\text{if }\lambda \le 1 \quad &\Rightarrow\quad 
  \tilde{O}\!\left(
    \max\left\{
      \tfrac{1}{\epsilon\alpha},\;
      \tfrac{1}{\alpha^2}
    \right\}
  \right)
\end{align*}
\end{theorem}

\begin{proof}
\emph{Privacy} follows from \Cref{lem:best-privacy}. For \emph{utility}, let
\(\mathcal{E}_{\text{est}}\) be the event that the coarse step attains a constant-factor estimate with probability at least \(1-\nicefrac{\beta}{3}\) (\Cref{lem:exp-quantile-utility} with target error $=\tfrac12$ and privacy budget \(\nicefrac{\epsilon}{3}\)). Under \(\mathcal{E}_{\text{est}}\), the decision matches the oracle choice whenever \(\lambda\not\in(\tfrac12,2)\), and in the middle region the two routes are within a constant factor anyway. Conditioned on the chosen branch, the accuracy \(\tilde\lambda\in[(1-\alpha)\lambda,(1+\alpha)\lambda]\) holds with probability at least \(1-\nicefrac{2\beta}{3}\) by \Cref{lem:mle-utility} for the MLE route (using budget \(\nicefrac{2\epsilon}{3}\)) or by \Cref{lem:exp-quantile-utility} for the quantile route (again with budget \(\nicefrac{2\epsilon}{3}\)). A union bound over the two stages yields overall success probability at least \(1-\beta\). The stated sample-size bound aggregates the dominating terms of the chosen branch.
\end{proof}


\section{Missing Proofs}
\label{apx-sec:missing-proofs}

\begin{fact}[\Cref{fact:tv-exp} restated]\label{apx-fact:tv-exp}
    Let $\mathrm{Exp}(\lambda_1), \mathrm{Exp}(\lambda_2)$ be two exponential distributions and assume $\lambda_1 \geq \lambda_2$. Denote $a = \frac{\ln(\lambda_1) - \ln(\lambda_2)}{\lambda_1 - \lambda_2}$, then:
    \[
    TV(\mathrm{Exp}(\lambda_1), \mathrm{Exp}(\lambda_2)) = e^{-\lambda_2 \cdot a} - e^{-\lambda_1 \cdot a}
    \]
\end{fact}
\begin{proof}
    By definition of total variation and the PDF of exponential distribution:
    \[
    TV(\mathrm{Exp}(\lambda_1), \mathrm{Exp}(\lambda_2)) = \frac{1}{2}\int_0^\infty |\lambda_1e^{-\lambda_1 x} - \lambda_2e^{-\lambda_2 x}|~ dx
    \]
    Recall that $\lambda_1 \geq \lambda_2$ without loss of generality, then let's compute the unique point where the two PDFs meet:
    \[
    \lambda_1e^{-\lambda_1 a} = \lambda_2e^{-\lambda_2a} \iff \ln(\lambda_1) - \lambda_1 a = \ln(\lambda_2) - \lambda_2 a \iff a = \frac{\ln(\lambda_1) - \ln(\lambda_2)}{\lambda_1 - \lambda_2}
    \]
    Now we compute:
    \begin{equation*}
    \begin{split}
    \int_0^\infty |\lambda_1e^{-\lambda_1 x} & - \lambda_2e^{-\lambda_2 x}|~ dx = \int_0^a (\lambda_1e^{-\lambda_1 x} - \lambda_2e^{-\lambda_2 x})~ dx~ - \int_a^\infty (\lambda_1e^{-\lambda_1 x} - \lambda_2e^{-\lambda_2 x})~ dx  \\
    & = (e^{-\lambda_2 x} - e^{-\lambda_1 x})\big|_0^a - (e^{-\lambda_2 x} - e^{-\lambda_1 x})\big|_a^\infty = 2(e^{-\lambda_2 a} - e^{-\lambda_1 a})
    \end{split}
    \end{equation*}
    Hence we get:
    \[
    TV(\mathrm{Exp}(\lambda_1), \mathrm{Exp}(\lambda_2)) = e^{-\lambda_2 \cdot a} - e^{-\lambda_1 \cdot a}
    \]
\end{proof}

\begin{corollary}[\Cref{cor:tv-exp} restated]\label{apx-cor:tv-exp}
    If the rate parameters satisfy $\lambda_1 \in (1 \pm \alpha)\lambda_2$, then
    \[
    TV\!\left(\mathrm{Exp}(\lambda_1), \mathrm{Exp}(\lambda_2)\right) \leq \alpha
    \]
\end{corollary}
\begin{proof}
    Assume without loss of generality that $\lambda_1 \geq \lambda_2$.
    We first consider the case where $\lambda_1 = (1+\delta)\lambda_2$ with $|\delta| \leq \alpha$.  
    From the fact above, the total variation distance is:
    \[
    TV(\mathrm{Exp}(\lambda_1), \mathrm{Exp}(\lambda_2)) 
    = e^{-\lambda_2 a} - e^{-\lambda_1 a}
    \]
    where 
    \[
    a = \frac{\ln(\lambda_1) - \ln(\lambda_2)}{\lambda_1 - \lambda_2}
      = \frac{\ln(1+\delta)}{\delta \lambda_2}
    \]
    Substituting into the formula, we get:
    \begin{align*}
        TV(\mathrm{Exp}(\lambda_1), \mathrm{Exp}(\lambda_2)) &= e^{ - \lambda_2 \cdot \frac{\ln(1+\delta)}{\delta \lambda_2} } 
             - e^{ -(1+\delta)\lambda_2 \cdot \frac{\ln(1+\delta)}{\delta \lambda_2} } \\
           &= e^{ - \frac{\ln(1+\delta)}{\delta} } 
             - e^{ -(1+\delta)\frac{\ln(1+\delta)}{\delta} } \\
           &= (1+\delta)^{-1/\delta} - (1+\delta)^{-(1+\delta)/\delta} \\
           &= (1+\delta)^{-1/\delta}\cdot(1 - \frac{1}{(1+\delta)})
    \end{align*}
    Since $\delta \leq \alpha$ and the expression is monotonically increasing in $\delta > 0$, we can upper bound:
    \[
    TV(\mathrm{Exp}(\lambda_1), \mathrm{Exp}(\lambda_2)) \leq (1+\alpha)^{-1/\alpha}\cdot(1 - \frac{1}{(1+\alpha)})
    \]
    Finally
    \[
    (1+\alpha)^{1/\alpha}\cdot(1 - \frac{1}{(1+\alpha)}) = \alpha \cdot (1+\alpha)^{-(1+\alpha)/\alpha}
    \]
    Now compute:
    \begin{align*}
        (1+\alpha)^{-(1+\alpha)/\alpha} &= e^{\ln((1+\alpha)^{-(1+\alpha)\alpha})} = e^{-\frac{1+\alpha}{\alpha}\ln(1+\alpha)} \\
           &\overset{(\ast)}{\leq} e^{-\frac{1+\alpha}{\alpha}\cdot\frac{\alpha}{1+\alpha}} = e^{-1} < 1
    \end{align*}
    where inequality $(\ast)$ follows from $\ln(1+\alpha) \geq \frac{\alpha}{1+\alpha}$.

    Now consider $\lambda_2 = (1-\delta)\lambda_1$ with $|\delta| \leq \alpha$. then:
    \[
    a = \frac{-\ln(1-\delta)}{\delta \lambda_1}
    \]
    Following the same way, substituting into the formula, we get:
    \begin{align*}
        TV(\mathrm{Exp}(\lambda_1), \mathrm{Exp}(\lambda_2)) &= e^{ (1-\delta)\lambda_1 \cdot \frac{\ln(1-\delta)}{\delta \lambda_1} } 
             - e^{ \lambda_1 \cdot \frac{\ln(1-\delta)}{\delta \lambda_1} } \\
           &= e^{ (1-\delta)\frac{\ln(1-\delta)}{\delta} } 
             - e^{ \frac{\ln(1-\delta)}{\delta} } \\
           &= (1-\delta)^{(1-\delta)/\delta} - (1-\delta)^{1/\delta} \\
           &= (1-\delta)^{1/\delta}\cdot(\frac{1}{1-\delta} - 1)
    \end{align*}
    Since $\delta \leq \alpha$ and the expression is monotonically decreasing in $\delta > 0$, we can upper bound:
    \begin{align*}
        TV(\mathrm{Exp}(\lambda_1), \mathrm{Exp}(\lambda_2)) & \leq \lim\limits_{\delta \to 0} (1-\delta)^{1/\delta}\cdot(\frac{1}{1-\delta} - 1) \\
           &= \lim\limits_{\delta \to 0} (1-\delta)^{(1-\delta)/\delta}\cdot\delta = e^{-1}\cdot\delta \leq \alpha
    \end{align*}
    
    Overall we conclude:
    \[
    TV(\mathrm{Exp}(\lambda_1), \mathrm{Exp}(\lambda_2)) \leq \alpha.
    \]
\end{proof}

\begin{lemma}[\Cref{lem:quantile-privacy} restated]\label{apx-lem:quantile-privacy}
    \Cref{alg:quantile} is $\epsilon$-differentially private.
\end{lemma}

\begin{proof}
Let neighboring datasets $P,P'$ differ in one record. For each $i\ge 0$ define the
1/n–sensitive counting query
\[
q_i(P)\;=\;\frac{1}{n}\,\big|\{x\in P:\ x<2^i\lambda_{\min}\}\big|
\qquad(\Delta=1/n)
\]
The algorithm draws a noisy threshold $\tilde T=(1-\theta)+\mathrm{Lap}(2\Delta/\epsilon)$
and, scanning $i=0,1,2,\ldots$, draws independent noisy query values
$\tilde q_i=q_i(P)+\mathrm{Lap}(4\Delta/\epsilon)$ and halts at the first index
with $\tilde q_i\ge \tilde T$, returning that index (equivalently, $2^i\lambda_{\min}$),
or $\perp$ if none occurs.

This is exactly the \emph{AboveThreshold} (a.k.a.\ the \emph{Sparse Vector Technique})
with sensitivity $\Delta=1/n$, a single positive report, and noise scales
$2\Delta/\epsilon$ for the threshold and $4\Delta/\epsilon$ for the queries.
By the standard SVT privacy guarantee, the mechanism that outputs only the first
index $i$ whose noisy query exceeds the noisy threshold is $\epsilon$-differentially
private, independent of the number of queries evaluated
(see, e.g., \citealp[Theorem~3.23]{Dwork2006}). The final reported value $2^i\lambda_{\min}$ is a
deterministic post-processing of this index and therefore preserves $\epsilon$-DP.
\end{proof}

\begin{lemma}[\Cref{lem:quantile-utility} restated]\label{apx-lem:quantile-utility}
    \Cref{alg:quantile} returns a $6$-approximation of the $(1-\theta)$-quantile of $P$ w.p. $\geq 1-\beta$, given that $\nicefrac{1}{10} \leq \theta \leq \nicefrac{9}{10}$ and:
    \[
    n \geq \max\left\{\frac{20}{\epsilon} \ln\left(\frac{4\log\left(\nicefrac{\lambda_{\max}}{\lambda_{\min}}\right)}{\beta}\right), 200\ln\left(\nicefrac{4}{\beta}\right)\right\}
    \]
\end{lemma}
\begin{proof}
    The algorithm performs a search over dyadic multiples of the input parameter $\lambda_{\min}$, and outputs the smallest value $2^i \cdot \lambda_{\min}$ such that the noisy empirical cumulative distribution function (CDF) at that point exceeds a noisy threshold $\tilde{T} = (1-\theta) + \textrm{Lap}(\frac{2}{\epsilon n})$.

    Write $F(x)=1-e^{-\lambda x}$ for the true CDF and $F_n(x)$ for the empirical CDF. Two error sources affect each comparison with the threshold $1-\theta$:
    \begin{enumerate}
        \item \textbf{Sampling error.} By the Dvoretzky–Kiefer–Wolfowitz (DKW) inequality,
        \[
            \Pr\!\Big[ \sup_x \big|F_n(x)-F(x)\big| > \tfrac{1}{20} \Big]
            \ \le\ 2e^{-2n(1/20)^2}.
        \]
        Requiring this to be $\le \beta/2$ yields
        \[
            n \ \ge\ \frac{1}{2(1/20)^2}\ln\!\Big(\tfrac{4}{\beta}\Big)
            \ =\ 200 \ln\!\Big(\tfrac{4}{\beta}\Big).
        \]
        This implies that the algorithm’s comparison of the empirical CDF to the threshold $(1 - \theta)$ is accurate up to $\nicefrac{1}{20}$ due to sampling error.
    
        \item \textbf{Noise error.} Each comparison uses (i) a Laplace perturbation of the empirical CDF and (ii) a Laplace perturbation of the threshold.
        Let $b=4/(\epsilon n)$ be the Laplace scale. For any $t>0$,
        \(
            \Pr\big[|\mathrm{Lap}(b)|>t\big] = e^{-t/b}.
        \)
        Note that the threshold noise uses scale $2/(\epsilon n)$, which is smaller; thus, if we ensure that the Laplace noise with scale $4/(\epsilon n)$ is bounded, the threshold noise is automatically bounded as well.
        We require every noise draw among at most $2\log(\nicefrac{\lambda_{\max}}{\lambda_{\min}})$ independent draws (one for the CDF at each of $\log(\nicefrac{\lambda_{\max}}{\lambda_{\min}})$ checkpoints and one for the threshold; the factor $2$ is a safe constant) to be bounded by $1/20$. By a union bound,
        \[
            2\log(\nicefrac{\lambda_{\max}}{\lambda_{\min}})\cdot e^{-(1/20)/b} \ \le\ \beta/2
            \quad\Longleftrightarrow\quad
            n \ \ge\ \frac{5}{\epsilon}\,\ln\!\Big(\tfrac{4\log(\nicefrac{\lambda_{\max}}{\lambda_{\min}})}{\beta}\Big).
        \]
    \end{enumerate}
    
    Under these two events, we have with probability at least $1 - \beta$, both the noisy empirical CDF and the threshold are within $\nicefrac{1}{20}$ of their true values. Consequently, since $q_{1-\theta} = \frac{1}{\lambda}\ln(\frac{1}{\theta})$ for the exponential distribution, the decision to return a given $2^i \cdot \lambda_{\min}$ is correct up to a deviation of $2\cdot \max\{\frac{\ln(\frac{1}{\theta - 3/20})}{\ln{\frac{1}{\theta}}}, \frac{\ln(\frac{1}{\theta})}{\ln(\frac{1}{\theta + 3/20})}\}$. In particular, if $\nicefrac{1}{10} \leq \theta \leq \nicefrac{9}{10}$ then the quantile scales by at most
    \[
    \text{max}
    \begin{cases}
    \displaystyle
    2\,\frac{\ln\!\left(10\right)}{\ln\!\left(4\right)}
    \approx 3.322, & \theta = \nicefrac{1}{4},\\[10pt]
    \displaystyle
    2\,\frac{\ln\!\left(\frac{4}{3}\right)}{\ln\!\left(\frac{10}{9}\right)}
    \approx 5.461, & \theta = \nicefrac{9}{10}.
    \end{cases}
    \]
    then a clean uniform bound is $\leq 6$.
    Thus, with probability at least $1-\beta$, the output $\tilde Q_{1-\theta}$ satisfies
    \[
    \frac{1}{6}q_{1-\theta} \;\le\; \tilde Q_{1-\theta} \;\le\; 6\,q_{1-\theta},
    \]
    i.e., Algorithm~\ref{alg:quantile} achieves a constant $6$-approximation to the $(1 - \theta)$-quantile of $P$, provided
    \[
    n \geq \max\left\{\frac{5}{\epsilon} \ln\left(\frac{4\log\left(\nicefrac{\lambda_{\max}}{\lambda_{\min}}\right)}{\beta}\right), 200\ln\left(\nicefrac{4}{\beta}\right)\right\}
    \]

\end{proof}

\begin{corollary}[\Cref{cor:range} restated]\label{apx-cor:range}
Let $\tilde Q_{1-\theta}$ be the output of \Cref{alg:quantile}. Fix a target failure probability $2\beta\in(0,1)$. There exists an absolute constant $c_0\ge 1$ from \Cref{lem:quantile-utility} such that, if we set
\[
\tilde R \;=\; C\,\tilde Q_{1-\theta}\,\ln n
\quad\text{with}\quad
C \;\ge\; \frac{c_0}{\ln(1/\theta)}\Big(1+\frac{\ln(1/\beta)}{\ln n}\Big),
\]
then with probability at least $1-2\beta$ all $n$ samples lie in $[0,\tilde R]$. In particular,
\[
\tilde R \;=\; O\!\left(\frac{1}{\lambda}\ln\!\Big(\nicefrac{1}{\theta}\Big)\,\ln n\right)
\]
\end{corollary}

\begin{proof}
Let $Q_{1-\theta}=\frac{1}{\lambda}\ln(\nicefrac{1}{\theta})$ be the $(1-\theta)$-quantile of $\mathrm{Exp}(\lambda)$. By \Cref{lem:quantile-utility}, there exists $c_0$ such that with probability at least $1-\beta$:
\[
\frac{1}{c_0}\,Q_{1-\theta}\ \le\ \tilde Q_{1-\theta}\ \le\ c_0\,Q_{1-\theta}
\]
Denote this good event as $\mathcal{E}$ and condition on it. For any sample $X\sim\mathrm{Exp}(\lambda)$,
\[
\Pr[X>\tilde R]
\;\le\; e^{-\lambda\tilde R}
\;\le\; \exp\!\Big(-\lambda\,C\ln n\cdot \tfrac{Q_{1-\theta}}{c_0}\Big)
\;=\; \exp\!\Big(-\tfrac{C}{c_0}\ln n\cdot \ln\tfrac{1}{\theta}\Big)
\;=\; n^{-\frac{C}{c_0}\ln(1/\theta)}
\]
A union bound over the $n$ samples gives
\[
\Pr\!\big[\max_i X_i > \tilde R\ \big|\ \mathcal{E}\big]
\ \le\ n^{\,1-\frac{C}{c_0}\ln(1/\theta)}
\]
Choosing
\(
C \ge \frac{c_0}{\ln(1/\theta)}\big(1+\frac{\ln(1/\beta)}{\ln n}\big)
\)
ensures the right-hand side is at most $\beta$. Unconditioning, the overall failure probability is at most $2\beta$.

Finally, under $\mathcal{E}$:
\[
\tilde R \;=\; C\,\tilde Q_{1-\theta}\,\ln n
\ \le\ C\,c_0\,Q_{1-\theta}\,\ln n
\ =\ O\!\left(\frac{1}{\lambda}\ln\!\Big(\nicefrac{1}{\theta}\Big)\,\ln n\right)
\]
which proves the claimed order bound.
\end{proof}

\begin{lemma}[\Cref{lem:mle-privacy} restated]\label{apx-lem:mle-privacy}
    \Cref{alg:private-mle} is $\epsilon$-differentially private.
\end{lemma}
\begin{proof}
    Let $f(P)=\frac{1}{n}\sum_{i=1}^n \min\{x_i,R\}$. For neighboring datasets $P,P'$ differing in one record, the $\ell_1$-sensitivity satisfies
    \[
    \Delta f\;\le\;\frac{1}{n}\,\big|\min\{x,R\}-\min\{x',R\}\big|\;\le\;\frac{R}{n}.
    \]
    Adding $\mathrm{Lap}(R/(\epsilon n))$ to $f(P)$ therefore yields an $\epsilon$-DP release by the Laplace mechanism. The final estimator $\tilde\lambda=1/\tilde s$ is a deterministic function (post-processing) of the noisy mean $\tilde s$, so the overall procedure is $\epsilon$-DP. 
\end{proof}

\begin{lemma}[\Cref{lem:mle-utility} restated]\label{apx-lem:mle-utility}
    \Cref{alg:private-mle} returns an $(1 \pm \alpha)$-approximation of $\lambda$ w.p. $\geq 1-\beta$, given that:
    \[
    n \geq O\left(\max\left\{\frac{R\ln(\frac{1}{\beta})}{\epsilon(\alpha - e^{-\lambda R})}, \frac{\ln(\frac{1}{\beta})}{\alpha^2}\right\}\right)
    \]
\end{lemma}
\begin{proof}
Let $X_1, \dots, X_n \overset{\text{i.i.d.}}{\sim} \text{Exp}(\lambda)$. The maximum likelihood estimate (MLE) of the rate parameter $\lambda$ is given by:
\[
    \hat{\lambda} = \left(\frac{1}{n} \sum_{i=1}^n X_i\right)^{-1}
\]

Let $\bar{s} = \frac{1}{n} \sum_{i=1}^n \bar{x}_i$ denote the mean of the clipped values, and $\tilde{s}$ the privatized version:
\[
    \tilde{s} = \bar{s} + \mathrm{Lap}\left(\frac{R}{\epsilon n}\right)
\]
Then, the private estimator is:
\[
    \tilde{\lambda} = \frac{1}{\tilde{s}}
\]

We now analyze the error in $\tilde{\lambda}$. First, we bound the deviation of $\tilde{s}$ from the true mean $s = 1/\hat\lambda$. The total error in $\tilde{s}$ arises from two sources:
\begin{enumerate}
    \item \textbf{Clipping Bias:} Define $\delta_{\text{clip}} = |\mathbb{E}[\bar{x}_i] - \mu|$. Since $X_i \sim \text{Exp}(\lambda)$, the clipping bias satisfies:
    \[
        \delta_{\text{clip}} = \int_R^\infty (x - R) \lambda e^{-\lambda x} dx = e^{-\lambda R} \left(\frac{1}{\lambda}\right)
    \]
    \item \textbf{Laplace Noise:} With probability at least $1 - \nicefrac{\beta}{2}$, the Laplace noise is bounded by:
    \[
        \left|\mathrm{Lap}\left(\frac{R}{\epsilon n}\right)\right| \leq \frac{R}{\epsilon n} \ln\left(\frac{2}{\beta}\right)
    \]
\end{enumerate}

Therefore, the total deviation of $\tilde{s}$ from $\mu$ is:
\[
    |\tilde{s} - s| \leq \delta_{\text{clip}} + \left|\mathrm{Lap}\left(\frac{R}{\epsilon n}\right)\right| \leq \frac{e^{-\lambda R}}{\lambda} + \frac{R}{\epsilon n} \ln\left(\frac{2}{\beta}\right)
\]

We now analyze the error in $\hat{\lambda}$. We apply a multiplicative Chernoff bound for the sample mean of i.i.d.\ exponential variables. For any \( \alpha \in (0,1) \), we have:
\[
\Pr\left[ \left| s - \frac{1}{\lambda} \right| \geq \frac{\alpha}{\lambda} \right] 
\leq 2 \exp\left( -\frac{n\alpha^2}{3} \right)
\]

Equivalently, since \( \hat{\lambda} = 1/s \), this implies:
\[
\Pr\left[ \left| \frac{\hat{\lambda}}{\lambda} - 1 \right| \geq \alpha \right] 
\leq 2 \exp\left( -\frac{n\alpha^2}{3} \right)
\]

To ensure this probability is at least \( 1 - \nicefrac{\beta}{2} \), it suffices to set:
\[
2 \exp\left( -\frac{n\alpha^2}{3} \right) \leq \nicefrac{\beta}{2} \implies n \geq \frac{3}{\alpha^2} \ln\left( \frac{4}{\beta} \right)
\]

Finally, to guarantee an $(1 \pm \alpha)$-multiplicative approximation to $\lambda$, it suffices to require:
\[
    \left| \tilde{s} - \frac{1}{\lambda} \right| \leq |\tilde s - s| + |s - \frac{1}{\lambda}| \leq \frac{\alpha}{\lambda}
\]
Thus, we set:
\[
\frac{e^{-\lambda R}}{\lambda} + \frac{R}{\epsilon n} \ln\left(\frac{2}{\beta}\right) \leq \frac{\alpha}{2\lambda}
\]

This gives the sample complexity condition:
\[
    n \geq O\left(\max\left\{\frac{R\ln(\frac{1}{\beta})}{\epsilon(\alpha - e^{-\lambda R})}, \frac{\ln(\frac{1}{\beta})}{\alpha^2}\right\}\right)
\]
under the assumption that $\alpha > e^{-\lambda R}$. Note that since $R = O\!\left(\tfrac{1}{\lambda}\ln(\tfrac{1}{\theta})\ln(n)\right)$ (see \Cref{cor:range}), this condition simplifies to $\alpha > 1/n$, which is a mild and natural requirement as typically $n > 1/\alpha$.

Hence, under this condition, the algorithm returns an $(1 \pm \alpha)$-approximation to $\lambda$ with probability at least $1 - \beta$.
\end{proof}

\begin{lemma}[\Cref{lem:exp-quantile-privacy} restated]\label{apx-lem:exp-quantile-privacy}
    Algorithm~\ref{alg:exp-quantile} is $\epsilon$-differentially private.
\end{lemma}
\begin{proof}
    Let $T=\left\lceil \log_2\!\left(\frac{\log(\lambda_{\max}/\lambda_{\min})}{\alpha}\right)\right\rceil$ be the maximal number of (binary-search) iterations. Each iteration computes one counting query of sensitivity $1/n$ and releases it with independent Laplace noise of scale $b=T/(\epsilon n)$. By the Laplace mechanism, each such release is $(\epsilon/T)$-DP; by sequential composition over at most $T$ iterations, the full vector of noisy comparisons is $\epsilon$-DP. The final output is a deterministic function of these noisy values, hence by post-processing the algorithm is $\epsilon$-DP. 
\end{proof}

\begin{lemma}[\Cref{lem:exp-quantile-utility} restated]\label{apx-lem:exp-quantile-utility}
    Fix target accuracy $\alpha \in (0,1)$ and failure probability $\beta \in (0,1)$. Let $T = \left\lceil \log_2\!\left(\frac{\log(\lambda_{\max}/\lambda_{\min})}{\alpha}\right) \right\rceil$. If
    \[
        n \ \ge\ \max\left\{
            \frac{2eT}{\epsilon \alpha} \ln\left( \frac{2T}{\beta} \right),
            \frac{2}{\alpha^2} \ln\left( \frac{2T}{\beta} \right)
        \right\}
    \]
    then with probability at least $1 - \beta$ Algorithm~\ref{alg:exp-quantile} outputs $\tilde{\lambda}$ satisfying
    \[
        \tilde{\lambda} \in \left[ (1-\alpha)\lambda,\ (1+\alpha)\lambda \right]
    \]
\end{lemma}
\begin{proof}
    The binary search compares $\tilde{p}$ with $1 - 1/e$ at each step. Two error sources affect $\tilde{p}$:

    \begin{enumerate}
        \item \textbf{Sampling error.} By the DKW inequality,
        \[
            \Pr\!\left[ \sup_x |F_n(x) - F(x)| > \frac{\alpha}{2} \right] \le 2 e^{-2n(\alpha/2)^2}
        \]
        Setting this $\le \beta/(2T)$ and union-bounding over $T$ steps gives the sampling term
        \[
            n \ge \frac{2}{\alpha^2} \ln\left( \frac{2T}{\beta} \right)
        \]

        \item \textbf{Noise error.}
        Recall that the objective of the quantile estimation step is to output a value within the interval $\left[ \frac{1}{(1+\alpha)\lambda},\ \frac{1}{(1-\alpha)\lambda} \right]$.
        It suffices to estimate the $1/e$-quantile to within an additive error of $\frac{\alpha}{2e}$ in the CDF domain. Indeed, consider the upper deviation:
        \[
            \int_{1/\lambda}^{1/((1-\alpha)\lambda)} \lambda e^{-\lambda x} \, dx
            = e^{-1} - e^{-1/(1-\alpha)}
            = e^{-1}\!\left(1 - e^{-\frac{\alpha}{1-\alpha}}\right)
        \]
        Using the inequality $e^{-t} \le 1 - t/2$ for $t \in [0,1]$  and noting that $\alpha < \nicefrac{1}{2}$, we obtain
        \[
            e^{-1}\!\left(1 - e^{-\frac{\alpha}{1-\alpha}}\right)
            \ge e^{-1}\!\left(1 - \left(1 - \frac{\alpha}{2(1-\alpha)}\right)\right)
            = \frac{\alpha}{2e(1-\alpha)}
            \ge \frac{\alpha}{2e}
        \]
        
        Similarly, for the lower deviation we have
        \[
            \int_{1/((1+\alpha)\lambda)}^{1/\lambda} \lambda e^{-\lambda x} \, dx
            = e^{-1/(1+\alpha)} - e^{-1}
            = e^{-1}\!\left(e^{\frac{\alpha}{1+\alpha}} - 1\right)
        \]
        Applying the bound $e^{t} \ge 1 + t$ for $t \in [0,1]$, we obtain
        \[
            e^{-1}\!\left(e^{\frac{\alpha}{1+\alpha}} - 1\right)
            \ge e^{-1}\!\left(\frac{\alpha}{1+\alpha}\right)
            \ge \frac{\alpha}{2e}
        \]
        
        Therefore, an additive $\frac{\alpha}{2e}$-accuracy in the estimated CDF value is sufficient to guarantee that the resulting quantile lies in the desired multiplicative interval for $1/\lambda$.
        
        Laplace noise with scale $b = T/(\epsilon n)$ satisfies $\Pr[|\mathrm{Lap}(b)| > \nicefrac{\alpha}{2e}] \le e^{-(\nicefrac{\alpha}{2e})/b}$. Setting this $\le \beta/(2T)$ yields
        \[
            n \ge \frac{2eT}{\epsilon \alpha} \ln\left( \frac{2T}{\beta} \right)
        \]
    \end{enumerate}
    
    When both bounds hold, at each step the comparison with $1 - 1/e$ is correct, ensuring that after $T$ steps the returned $\tilde{q}$ lies in
    \[
        \left[ \frac{1}{(1+\alpha)\lambda},\ \frac{1}{(1-\alpha)\lambda} \right]
    \]
    Inverting yields $\tilde{\lambda} \in [(1-\alpha)\lambda, (1+\alpha)\lambda]$.
\end{proof}

\begin{lemma}[\Cref{lem:best-privacy} restated]\label{apx-lem:best-privacy}
Algorithm~\ref{alg:best-of-both} is $\epsilon$-differentially private.
\end{lemma}
\begin{proof}
The algorithm proceeds in two stages:
(i) it runs the quantile-based estimator with budget $\epsilon/3$ and obtains $\tilde\lambda_0$; (ii) conditioned on this \emph{released} value, it runs either the MLE route or the quantile route with fresh budget $2\epsilon/3$. By the post-processing property, choosing the branch based on $\tilde\lambda_0$ does not incur additional privacy loss, and by adaptive sequential composition the total privacy loss is at most $\epsilon/3+2\epsilon/3=\epsilon$. 
\end{proof}

\begin{lemma}[\Cref{lem:bounds} restated]\label{apx-lem:bounds}
    \Cref{alg:median-dp} is $(\epsilon, \delta)$-differentially private. Moreover, if
    \[
    n \geq \max\left\{\frac{800}{\epsilon}\ln(\frac{2}{\delta\beta}), 5000\ln(\frac{2}{\beta})\right\}
    \]
    then
    \[
    \Pr[2^{k^* - 1} < \nicefrac{\ln(2)}{\lambda} < 2^{k^* + 1}] \geq 1 - \beta
    \]
\end{lemma}

\begin{proof}
    A proof of $(\epsilon, \delta)$-differential privacy of this algorithm can be found in Theorem 3.5 of~\cite{Vadhan2017}. For utility, we show that if
    \[
    n \geq \max\left\{\frac{800}{\epsilon}\ln(\frac{2}{\delta\beta}), 5000\ln(\frac{2}{\beta})\right\}
    \]
    then $\Pr[2^{k^* - 1} < \nicefrac{\ln(2)}{\lambda} < 2^{k^* + 1}] \geq 1 - \beta$. Let $m=\nicefrac{\ln(2)}{\lambda}$ be the population median and let $p_k=\Pr[X_i\in(2^k,2^{k+1})]=e^{-\lambda 2^k}-e^{-\lambda 2^{k+1}}$. The map
    $k\mapsto p_k$ is unimodal and maximized at indices $k$ with $\lambda 2^k$ closest
    to $\ln(2)$, hence the top two masses are at $k_m=\lfloor\log_2 m\rfloor$ and $k_m+1$,
    with $p_k\in[0.207,0.25]$ by a direct calculation. Lemma 2.3 of~\cite{Karwa2017FiniteSD} gives that if $n \geq \max\{\frac{800}{\epsilon}\ln(\frac{2}{\delta\beta}), 5000\ln(\frac{2}{\beta})\}$ then $\Pr[|\tilde p_{k_m} - p_{k_m}| > \nicefrac{1}{100}] \leq \beta$. Laplace noise with scale
    $2/(\varepsilon n)$ and the threshold $\frac{2}{\epsilon n}\ln(\nicefrac{2}{\delta}) + \frac{1}{n}$ ensure that (i) no empty bin enters $S$ with probability $1-\delta$, and (ii) the perturbation cannot change the winner
    outside $\{k_m,k_m+1\}$. Therefore $k^*\in\{k_m,k_m+1\}$ with probability at least $1 - \beta$, and the expanded interval
    $(2^{k^*-1},2^{k^*+1})$ contains $(2^{k_m},2^{k_m+1})\ni m$. 
    Rearranging we have:
    \begin{equation*}
    \begin{split}
    & 2^{k^* - 1} < \nicefrac{\ln(2)}{\lambda} < 2^{k^* + 1} \iff (2^{k^* + 1})^{-1} < \nicefrac{\lambda}{\ln(2)} < (2^{k^* - 1})^{-1} \\
    & \quad \iff \ln(2)(2^{k^* + 1})^{-1} < \lambda < \ln(2)(2^{k^* - 1})^{-1} \iff \tilde \lambda_{\min} < \lambda < \tilde \lambda_{\max}
    \end{split}
    \end{equation*}
\end{proof}



\begin{lemma}[\Cref{lem:tv-alpha-exp} restated]\label{apx-lem:tv-alpha-exp}
Let $T(r)$ denote the total variation distance between two exponential distributions whose
rates differ by a ratio $r\ge 1$:
\begin{align*}
T(r) &:= \operatorname{TV}\!\big(\mathrm{Exp}(r\lambda),\mathrm{Exp}(\lambda)\big) \notag = r^{-\tfrac{1}{r-1}} - r^{-\tfrac{r}{r-1}} \notag \\
     &= r^{-\tfrac{1}{r-1}}\!\left(1 - \tfrac{1}{r}\right)
\end{align*}
Fix any $\alpha\in(0,\tfrac12)$. Then there exists an absolute constant $c$ such that,
with $r=1+c\alpha$, one has $T(r)\ge \alpha$. In particular,
the choice $c=8$ suffices.
\end{lemma}

\begin{proof}
Write $r=1+\delta$ with $\delta>0$ and define
\begin{align*}
f(\delta) &:= T(1+\delta)
   = \frac{\delta}{1+\delta}\,(1+\delta)^{-1/\delta}, \\[6pt]
R(\delta) &:= \frac{f(\delta)}{\delta}
   = \frac{(1+\delta)^{-1/\delta}}{1+\delta}
\end{align*}
We prove that $R$ is strictly decreasing on $(0,\infty)$. Indeed,
\[
\ln R(\delta)\;=\;-\frac{\ln(1+\delta)}{\delta}-\ln(1+\delta)
\]
so
\begin{align*}
\frac{\mathrm d}{\mathrm d\delta}\ln R(\delta)
&= -\Bigg(\frac{1}{\delta(1+\delta)}
    - \frac{\ln(1+\delta)}{\delta^2}\Bigg)
    - \frac{1}{1+\delta} \\[6pt]
&= \frac{(1+\delta)\ln(1+\delta)-\delta-\delta^2}
        {\delta^2(1+\delta)}
\end{align*}
Let $\psi(\delta):=(1+\delta)\ln(1+\delta)-\delta-\delta^2$. Then
\[
\psi'(\delta)=\ln(1+\delta)-2\delta<\delta-2\delta=-\delta<0\quad(\delta>0)
\]
because $\ln(1+x)<x$ for $x>0$. Since $\psi(0)=0$ and $\psi'$ is negative on $(0,\infty)$, we have $\psi(\delta)<0$ for all $\delta>0$. Hence $(\ln R)'(\delta)<0$, so $R$ is strictly decreasing.

Consequently, for every $\delta\in(0,4]$,
\[
R(\delta)\;\ge\;R(4)
=\frac{(1+4)^{-1/4}}{1+4}
=\frac{1}{5^{5/4}}
\]
Now fix $c=8$ and take any $\alpha\in(0,\tfrac12)$. Then $\delta=c\alpha\in(0,4)$ and
\[
T(1+c\alpha)\;=\;f(\delta)\;=\;\delta\,R(\delta)
\;\ge\;\delta\,R(4)
\;=\;\frac{c}{5^{5/4}}\;\alpha
\]
Since $8\ge 5^{5/4}$, we obtain
$T(1+8\alpha)\ge \alpha$ for all $\alpha\in(0,\tfrac12)$, as required.
\end{proof}

\section{Applications to Pareto Distributions}
\label{apx-sec:pareto}

The exponential distribution is not only fundamental in its own right but also serves as a key building block for other families. 
In particular, the \emph{Pareto distribution}---one of the most widely studied heavy-tailed laws in economics, finance, and networks---admits a simple logarithmic reduction to the exponential. 
This reduction allows us to transfer our private exponential learning algorithms directly to the Pareto setting.

\subsection{Learning a Pareto Distribution}
\label{subsec:pareto}

The techniques developed for exponential distribution can also be adapted to privately learn the parameters of a \emph{Pareto distribution}, which models heavy-tailed phenomena. Recall that the standard Pareto distribution, that is $\text{Pareto}(x_m, \alpha_p)$, is defined over $x \geq x_m > 0$ with probability density function:
\[
f(x) = \alpha_p \cdot \frac{x_m^{\alpha_p}}{x^{\alpha_p+1}} \qquad \text{for } x \geq x_m,
\]
where $x_m$ is the minimum value and $\alpha_p > 0$ is the shape parameter.

We consider the case where $x_m$ is \emph{known and fixed}, and the goal is to estimate the shape parameter $\alpha_p$. The MLE estimator is:
\[
\hat{\alpha_p} = \left(\frac{1}{n} \sum_{i=1}^n \ln\left(\frac{x_i}{x_m}\right) \right)^{-1}
\]
This estimator closely resembles the exponential MLE, since $\ln(x/x_m) \sim \text{Exp}(\alpha_p)$ when $x \sim \text{Pareto}(x_m, \alpha_p)$. This logarithmic transformation allows us to reduce the problem of learning a Pareto distribution to that of learning an exponential distribution, which we already handle.

To learn $\alpha_p$ under differential privacy:
\begin{enumerate}
    \item \textbf{Log-transform the data:} For each $x_i \in P$, define $y_i = \ln(x_i / x_m)$. Then $y_i \sim \text{Exp}(\alpha_p)$.
    \item \textbf{Apply \Cref{alg:exp-mle}:} Use our private exponential learning algorithm on the transformed dataset $\{y_i\}_{i=1}^n$ to obtain a private estimate $\tilde \alpha_p$.
\end{enumerate}

\begin{lemma}\label{lem:pareto-tv-gap}
Let $P(x_m,\alpha_p)$ denote the Pareto distribution.
For two parameter pairs $(x_m,\alpha_p)$ and $(\tilde x_m,\tilde\alpha_p)$ set
\[
\Delta_s:=\big|\ln(\tilde x_m/x_m)\big|,\qquad
\alpha_{\min}:=\min\{\alpha_p,\tilde\alpha_p\},\ \ \alpha_{\max}:=\max\{\alpha_p,\tilde\alpha_p\}
\]
Then the total variation distance satisfies
\[
\mathrm{TV}\!\left(P(x_m,\alpha_p),\,P(\tilde x_m,\tilde\alpha_p)\right)
\ \le\
1-e^{-\alpha_{\max}\Delta_s}
\ +\
\sqrt{\tfrac12\!\left(\tfrac{\alpha_{\max}}{\alpha_{\min}}-1-\ln\tfrac{\alpha_{\max}}{\alpha_{\min}}\right)}\!
\]
and, for small gaps, the linearized bound
$\mathrm{TV}\ \le\ \alpha_{\max}\Delta_s+\frac{|\tilde\alpha_p-\alpha_p|}{2\alpha_{\min}}$ also holds.
\end{lemma}

\begin{proof}
By the triangle inequality in total variation,
\[
\mathrm{TV}\big(P(x_m,\alpha_p),P(\tilde x_m,\tilde\alpha_p)\big)
\le \mathrm{TV}\big(P(x_m,\alpha_p),P(\tilde x_m,\alpha_p)\big)
   +\mathrm{TV}\big(P(\tilde x_m,\alpha_p),P(\tilde x_m,\tilde\alpha_p)\big)
\]
The first term compares equal-tail-index Paretos and depends only on scale.
For a fixed tail index $\bar\alpha$, the exact formula is
$\mathrm{TV}\big(P(x_m,\bar\alpha),P(\tilde x_m,\bar\alpha)\big)
=1-\exp(-\bar\alpha\,|\ln(\tilde x_m/x_m)|)$;
taking $\bar\alpha=\alpha_{\max}$ gives the stated upper bound
$1-e^{-\alpha_{\max}\Delta_s}$.

For the second term (equal scale $\tilde x_m$), apply Pinsker’s inequality
$\mathrm{TV}\le\sqrt{\tfrac12\,\mathrm{KL}}$ together with the closed-form KL
between equal-scale Paretos:
\[
\mathrm{KL}\!\left(P(\tilde x_m,\alpha_p)\,\Vert\,P(\tilde x_m,\tilde\alpha_p)\right)
=\frac{\alpha_p}{\tilde\alpha_p}-1-\ln\frac{\alpha_p}{\tilde\alpha_p}
\]
By symmetry in $(\alpha_p,\tilde\alpha_p)$ this equals
$\frac{\alpha_{\max}}{\alpha_{\min}}-1-\ln\frac{\alpha_{\max}}{\alpha_{\min}}$,
and the stated bound follows. The linearized bound is the first-order expansion
$1-e^{-t}\le t$ and $u-1-\ln u\le \tfrac12(u-1)^2$ near $t=0$, $u=1$.
\end{proof}

\begin{corollary}\label{cor:pareto-tv}
Let $P(x_m, \alpha_p)$ and $P(\tilde x_m, \tilde\alpha_p)$ be two Pareto distributions. If
\[
\tilde\alpha_p \in \big[(1-\gamma)\alpha_p,\ (1+\gamma)\alpha_p\big] \quad \text{and} \quad \frac{\tilde x_m}{x_m} \leq e^{\frac{\gamma}{2\alpha_p}}
\]
then $\mathrm{TV}\!\left(P(x_m,\alpha_p),\,P(\tilde x_m,\tilde\alpha_p)\right)\ \le\ \gamma$.

\end{corollary}

The reduction from Pareto to exponential distributions allows us to directly transfer our best-of-both exponential learner into the Pareto setting. 
By applying the logarithmic transform to the input samples, running the exponential algorithm, and then inverting the transform, we obtain an efficient private learner for Pareto distributions. Recall that our private exponential learner requires known bounds $0 < \lambda_{\min} < \lambda_{\max}$ on the exponential rate parameter $\lambda$. Under the logarithmic reduction $y=\ln(x/x_m)$, the Pareto shape parameter $\alpha_p$ appears \emph{exactly} as this exponential rate. Consequently, throughout this section the quantities $\alpha_{p_{\min}}$ and $\alpha_{p_{\max}}$ should be interpreted as known
bounds on the Pareto parameter $\alpha_p$.
The following corollary summarizes the resulting guarantees.

\begin{corollary}
    Given $n$ samples from $\text{Pareto}(x_m, \alpha_p)$, the composition of the log-transform with \Cref{alg:exp-mle} yields an $\epsilon$-differentially private estimate $\tilde \alpha_p$ such that, with probability at least $1 - \beta$,
    \[
        |\tilde \alpha_p - \alpha_p| \leq \gamma\alpha_p
    \]
    for any $\gamma \in (0,1)$, provided that
    \[
    \begin{aligned}
    n \geq O\Bigg(\max\Bigg\{&
    \min\Bigg\{\frac{\ln\!\left(\frac{\log(\nicefrac{\alpha_{p_{\max}}}{\alpha_{p_{\min}}})}{\alpha}\right)}{\epsilon\gamma}
    \ln\!\left(\frac{\ln(\nicefrac{\log(\nicefrac{\alpha_{p_{\max}}}{\alpha_{p_{\min}}})}{\alpha})}{\beta}\right),
    \frac{\ln(\frac{1}{\alpha})\ln(\frac{1}{\beta})}{\alpha_p\epsilon\gamma}
    \Bigg\}, \\
    &\frac{\ln\!\left(\frac{\ln(\nicefrac{\log(\nicefrac{\alpha_{p_{\max}}}{\alpha_{p_{\min}}})}{\alpha})}{\beta}\right)}{\gamma^2},
    \frac{\ln(\frac{\log(\nicefrac{\alpha_{p_{\max}}}{\alpha_{p_{\min}}})}{\beta})}{\epsilon}
    \Bigg\}\Bigg)
    \end{aligned}
    \]
    which yields $\mathrm{TV}\!\left(P(x_m,\alpha_p),\,P(x_m,\tilde\alpha_p)\right)\ \le\ \gamma$ by \Cref{cor:pareto-tv}.
\end{corollary}

\subsubsection{Estimating the Scale Parameter $x_m$ Under Differential Privacy}
\label{subsec:pareto-xm}

We now address the case where both the shape $\alpha_p$ and the scale $x_m$ of a Pareto distribution $\mathrm{Pareto}(x_m,\alpha_p)$ are unknown. A direct private release of the sample minimum is problematic due to its large global sensitivity. Instead, we leverage a \emph{quantile-tail} reduction that preserves both privacy and statistical efficiency.

\paragraph{Key observation (scale invariance).}
If $X \sim \mathrm{Pareto}(x_m,\alpha_p)$ and $t \ge x_m$, then conditional on $X \ge t$ we have the exact scaling law
\[
\left.\ln\!\left(\frac{X}{t}\right)\ \right|\ \{X \ge t\}\ \sim\ \mathrm{Exp}(\alpha_p)
\]
\begin{lemma}[Tail-log transform of Pareto is exponential]\label{lem:pareto-tail-exp}
Let $X\sim \mathrm{Pareto}(x_m,\alpha_p)$ with CDF $F(x)=1-(x_m/x)^{\alpha_p}$ for $x\ge x_m$, and fix any $t\ge x_m$. Then
\[
\left.\ln\!\Big(\frac{X}{t}\Big)\ \right|\ \{X\ge t\}\ \sim\ \mathrm{Exp}(\alpha_p)
\]
\end{lemma}

\begin{proof}
For $y\ge 0$,
\[
\Pr\!\left[\left.\ln\!\Big(\frac{X}{t}\Big)>y\ \right|\ X\ge t\right]
\;=\;
\frac{\Pr\big[\ln(X/t)>y,\ X\ge t\big]}{\Pr[X\ge t]}
\;=\;
\frac{\Pr\big[X>t e^{y}\big]}{\Pr[X\ge t]}
\]
Since $X$ is Pareto,
\[
\Pr[X>u]=\left(\frac{x_m}{u}\right)^{\alpha_p}\qquad \text{for }u\ge x_m.
\]
Applying this with $u=t e^{y}$ and with $u=t$ gives
\[
\Pr\!\left[\left.\ln\!\Big(\frac{X}{t}\Big)>y\ \right|\ X\ge t\right]
\;=\;
\frac{(x_m/(t e^{y}))^{\alpha_p}}{(x_m/t)^{\alpha_p}}
\;=\;
e^{-\alpha_p y}
\]
Thus the conditional tail function is $e^{-\alpha_p y}$ for $y\ge 0$, which is exactly the tail of $\mathrm{Exp}(\alpha_p)$. Hence 
\(
\left.\ln(X/t)\ \right| \{X\ge t\}\sim \mathrm{Exp}(\alpha_p)
\)
\end{proof}

Thus, if we can privately estimate a low quantile $q_\tau$ (for a small fixed $\tau \in (0,1)$) and then look only at exceedances above $q_\tau$, the log-exceedances behave as i.i.d.\ exponentials with rate $\alpha_p$, independently of the unknown $x_m$.

\begin{lemma}[Pareto Quantile and Scale Identity]\label{lem:pareto-quantile}
Let $X \sim \mathrm{Pareto}(x_m,\alpha_p)$ with CDF $F(x)=1-(x_m/x)^{\alpha_p}$ for $x\ge x_m$, where $x_m>0$ and $\alpha_p>0$. For any $\tau\in(0,1)$, the $\tau$-quantile $q_\tau$ satisfies
\[
q_\tau \;=\; \frac{x_m}{(1-\tau)^{1/\alpha_p}}.
\]
Equivalently,
\[
x_m \;=\; q_\tau\,(1-\tau)^{1/\alpha_p}.
\]
\end{lemma}

\begin{proof}
By definition of the $\tau$-quantile, $q_\tau$ is any value such that $F(q_\tau)=\tau$. For $X\sim\mathrm{Pareto}(x_m,\alpha_p)$ and $q_\tau\ge x_m$,
\[
F(q_\tau) \;=\; 1 - \Big(\frac{x_m}{q_\tau}\Big)^{\alpha_p} \;=\; \tau
\implies
\Big(\frac{x_m}{q_\tau}\Big)^{\alpha_p} \;=\; 1-\tau.
\]
Taking $1/\alpha_p$ powers (noting $\alpha_p>0$) yields
\[
\frac{x_m}{q_\tau} \;=\; (1-\tau)^{1/\alpha_p}
\implies
q_\tau \;=\; \frac{x_m}{(1-\tau)^{1/\alpha_p}}.
\]
Rearranging gives the equivalent identity $x_m=q_\tau(1-\tau)^{1/\alpha_p}$. 
\end{proof}

\begin{remark}
Lemma~\ref{lem:pareto-quantile} is just the closed-form quantile function $F^{-1}(\tau)$ of the Pareto distribution. It underlies our scale recovery step: given (private) estimates $(\hat q_\tau,\tilde\alpha_p)$, we set $\tilde x_m=\hat q_\tau (1-\tau)^{1/\tilde\alpha_p}$.
\end{remark}

Having established the logarithmic reduction from Pareto to exponential distributions, we can now lift our adaptive best-of-both exponential learner directly to the Pareto setting. 
The idea is simple: first apply the $\log$-transform to convert Pareto samples into exponential ones, then run the best-of-both exponential learner under differential privacy, and finally invert the transformation to recover a private estimate of the Pareto parameter. 
The complete procedure is summarized in \Cref{alg:dp-pareto-both}.

\begin{algorithm}[H]
\caption{DP Learning of $\alpha_p$ and $x_m$ for Pareto$(x_m,\alpha_p)$}
\label{alg:dp-pareto-both}
\begin{algorithmic}
\State \hspace*{-\algorithmicindent} \textbf{Input:} $P=\{x_i\}_{i=1}^n$, bounds $0 < \alpha_{p_{\min}} < \alpha_{p_{\max}}$, target error $\alpha\in(0,1)$, privacy budget $\epsilon>0$.
\State \hspace*{-\algorithmicindent} \textbf{Output:} $(\tilde{\alpha}_p,\tilde{x}_m)$.
\State $\tau \gets \frac{1}{4\ln(7)}$
\State $\hat{q}_\tau \gets \Call{Quantile-Estimation}{P,\ \lambda_{\min} = \alpha_{p_{\min}},\ \lambda_{\max} = \alpha_{p_{\max}},\ \theta=1-\tau,\ \nicefrac{\epsilon}{2}}$ \Comment{DP low-quantile}
\State Let $I=\{i: x_i\ge \hat{q}_\tau\}$ and define $y_i=\ln(x_i/\hat{q}_\tau)$ for $i\in I$. \Comment{Tail transform}
\State $\tilde \alpha_p \gets \Call{Best-of-Both}{\{y_i\}_{i\in I},\ \lambda_{\min} = \hat{q}_\tau,\ \lambda_{\max} = \alpha_{p_{\max}},\ \alpha,\ \nicefrac{\epsilon}{2}}$
\State \textbf{Scale recovery:}
\[
\tilde{x}_m \ \gets\ \hat{q}_\tau\cdot (1-\tau)^{1/\tilde{\alpha}_p}
\qquad\text{since}\qquad
q_\tau \ =\ \frac{x_m}{(1-\tau)^{1/\alpha_p}}.
\]
\State \Return $(\tilde{\alpha}_p,\tilde{x}_m)$
\end{algorithmic}
\end{algorithm}

\begin{lemma}[Privacy of~\Cref{alg:dp-pareto-both}]\label{lem:pareto-both-privacy}
    \Cref{alg:dp-pareto-both} is $\epsilon$-differentially private.
\end{lemma}
\begin{proof}
    The algorithm first runs the DP low-quantile routine with budget $\epsilon/2$; this is $\epsilon/2$-DP. It then forms the (deterministic) transformed dataset $\{y_i=\ln(x_i/\hat q_\tau): x_i\ge \hat q_\tau\}$ using the already released $\hat q_\tau$ and the raw data, and applies a second DP estimator (the \textsc{Best-of-Both} exponential learner) with budget $\epsilon/2$ to $\{y_i\}$. For any fixed (public) $\hat q_\tau$, this second mechanism is $\epsilon/2$-DP with respect to the underlying dataset. Hence, by adaptive sequential composition the pair of releases is $\epsilon$-DP. The final scale estimate $\tilde x_m=\hat q_\tau (1-\tau)^{1/\tilde\alpha_p}$ is a deterministic function of the two released values, so by post-processing the overall algorithm is $\epsilon$-DP. 
\end{proof}

\begin{lemma}[Utility of~\Cref{alg:dp-pareto-both}]\label{lem:pareto-both-utility} Fix $\tau\in(0,1/4]$ and target error $\gamma\in(0,1/2]$ and $\beta\in(0,1)$. If
\[
\begin{aligned}
n \ge \frac{1}{1-\tau} \cdot O\Bigg(\max\Bigg\{&
\min\Bigg\{\frac{\ln\!\left(\frac{\log(\nicefrac{\alpha_{p_{\max}}}{\alpha_{p_{\min}}})}{\alpha}\right)}{\epsilon\gamma}
\ln\!\left(\frac{\ln(\nicefrac{\log(\nicefrac{\alpha_{p_{\max}}}{\alpha_{p_{\min}}})}{\alpha})}{\beta}\right),
\frac{\ln(\frac{1}{\alpha})\ln(\frac{1}{\beta})}{\alpha_p\epsilon\gamma}
\Bigg\}, \\
&\frac{\ln\!\left(\frac{\ln(\nicefrac{\log(\nicefrac{\alpha_{p_{\max}}}{\alpha_{p_{\min}}})}{\alpha})}{\beta}\right)}{\gamma^2},
\frac{\ln(\frac{\log(\nicefrac{\alpha_{p_{\max}}}{\alpha_{p_{\min}}})}{\beta})}{\epsilon}
\Bigg\}\Bigg)
\end{aligned}
\]
then with probability at least $1-\beta$ the outputs of \Cref{alg:dp-pareto-both} satisfy
\[
\tilde{\alpha}_p \in \big[(1-\gamma)\alpha_p,\ (1+\gamma)\alpha_p\big]
\quad\text{and}\quad
\frac{\tilde{x}_m}{x_m}
\leq e^{\ \,2\ln(7)\,(\gamma/\alpha_p)\tau}
\]
for a universal constant $c>0$. In particular, choosing $\tau = \frac{1}{4\ln(7)}$ yields a sufficient estimate of $x_m$ which yields $\mathrm{TV}\!\left(P(x_m,\alpha_p),\,P(\tilde x_m,\tilde\alpha_p)\right)\ \le\ \gamma$ by \Cref{cor:pareto-tv}.
\end{lemma}

\begin{proof}
    
Write $n_\tau:=|I|$ for the (random) number of exceedances above $\hat{q}_\tau$. In expectation, $n_\tau \approx (1-\tau)n$. The accuracy of $(\tilde{\alpha}_p,\tilde{x}_m)$ follows from three components:

\begin{enumerate}
\item \textbf{Low-quantile accuracy.} By \Cref{lem:quantile-utility} with $\theta=1-\tau$ and confidence $\beta/2$, the private estimate $\hat{q}_\tau$ is a multiplicative $O(1)$-approximation to $q_\tau$ provided
\[
n \geq \max\left\{\frac{20}{\epsilon} \ln\left(\frac{4\log\left(\nicefrac{\alpha_{p_{\max}}}{\alpha_{p_{\min}}}\right)}{\beta}\right), 200\ln\left(\nicefrac{4}{\beta}\right)\right\}
\]
\item \textbf{Rate accuracy on exceedances.} Conditioned on $\hat{q}_\tau$, the variables $\{y_i\}_{i\in I}$ are i.i.d.\ $\mathrm{Exp}(\alpha_p)$. Applying \Cref{thm:best-main} with sample size $n_\tau$ and confidence $\beta/2$ yields a multiplicative $(1\pm \gamma)$ estimate $\tilde{\alpha}_p$ when
\[
\begin{aligned}
n_\tau \ge C \cdot \max\Bigg\{&
\min\Bigg\{\frac{\ln\!\left(\frac{\log(\nicefrac{\alpha_{p_{\max}}}{\alpha_{p_{\min}}})}{\alpha}\right)}{\epsilon\gamma}
\ln\!\left(\frac{\ln(\nicefrac{\log(\nicefrac{\alpha_{p_{\max}}}{\alpha_{p_{\min}}})}{\alpha})}{\beta}\right),
\frac{\ln(\frac{1}{\alpha})\ln(\frac{1}{\beta})}{\alpha_p\epsilon\gamma}
\Bigg\}, \\
&\frac{\ln\!\left(\frac{\ln(\nicefrac{\log(\nicefrac{\alpha_{p_{\max}}}{\alpha_{p_{\min}}})}{\alpha})}{\beta}\right)}{\gamma^2},
\frac{\ln(\frac{\log(\nicefrac{\alpha_{p_{\max}}}{\alpha_{p_{\min}}})}{\beta})}{\epsilon}
\Bigg\}
\end{aligned}
\]
for a universal constant $C$.
\item \textbf{Scale recovery error propagation.} Since $x_m = q_\tau\,(1-\tau)^{1/\alpha_p}$ and $\tilde{x}_m=\hat{q}_\tau\,(1-\tau)^{1/\tilde{\alpha}_p}$,
\[
\frac{\tilde{x}_m}{x_m}
\ =\
\underbrace{\frac{\hat{q}_\tau}{q_\tau}}_{\text{quantile factor}}
\cdot
\underbrace{(1-\tau)^{\,\frac{1}{\tilde{\alpha}_p}-\frac{1}{\alpha_p}}}_{\text{rate factor}}
\ \leq\
7\cdot
\exp\!\left(\Big(\tfrac{1}{\tilde{\alpha}_p}-\tfrac{1}{\alpha_p}\Big)\!\ln(1-\tau)\right)
\]
If $\tilde{\alpha}_p\in[(1-\gamma)\alpha_p,(1+\gamma)\alpha_p]$, then
\[
\left|\tfrac{1}{\tilde{\alpha}_p}-\tfrac{1}{\alpha_p}\right|
\ \le\
\frac{\gamma}{(1-\gamma)\,\alpha_p}
\ \le\ \frac{2\gamma}{\alpha_p}
\quad\text{for }\gamma\le\tfrac12,
\]
and hence
\[
(1-\tau)^{\,\frac{1}{\tilde{\alpha}_p}-\frac{1}{\alpha_p}}
\leq 
\exp\!\Big(\tfrac{2\gamma}{\alpha_p}\,|\ln(1-\tau)|\Big)
\]
For $\tau\le \tfrac14$, $|\ln(1-\tau)|\le \tfrac{4}{3}\tau$. Thus,
\[
\frac{\tilde{x}_m}{x_m}
\in
7\cdot\exp\!\Big(\tfrac{2\gamma}{\alpha_p}\,\tau\Big)
\]
\end{enumerate}

\end{proof}

We now state the formal guarantees for the private Pareto learner. 
Since Algorithm~\ref{alg:dp-pareto-both} relies only on the logarithmic reduction and then invokes the exponential best-of-both learner, its privacy and accuracy properties follow directly from our earlier results. 
The theorem below shows that the algorithm achieves a $(1\pm \gamma)$ approximation to the Pareto parameter under differential privacy, which directly yields a $\gamma$-bound on the total variation distance between the learned and the true distributions, with essentially the same sample complexity as in the exponential case up to logarithmic factors.

\begin{theorem}
\label{thm:pareto-both-main}
Fix $\tau \in (0,1/4]$, target accuracy parameter $\gamma \in (0,1/2]$, failure probability $\beta \in (0,1)$, and privacy budget $\epsilon > 0$. 
Then Algorithm~\ref{alg:dp-pareto-both} is $\epsilon$-differentially private.

Moreover, if the sample size $n$ satisfies
\[
\begin{aligned}
n \ge \frac{C}{1-\tau} \cdot \max\Bigg\{&
\min\Bigg\{\frac{\ln\!\left(\frac{\log(\nicefrac{\alpha_{p_{\max}}}{\alpha_{p_{\min}}})}{\alpha}\right)}{\epsilon\gamma}
\ln\!\left(\frac{\ln(\nicefrac{\log(\nicefrac{\alpha_{p_{\max}}}{\alpha_{p_{\min}}})}{\alpha})}{\beta}\right),
\frac{\ln(\frac{1}{\alpha})\ln(\frac{1}{\beta})}{\alpha_p\epsilon\gamma}
\Bigg\}, \\
&\frac{\ln\!\left(\frac{\ln(\nicefrac{\log(\nicefrac{\alpha_{p_{\max}}}{\alpha_{p_{\min}}})}{\alpha})}{\beta}\right)}{\gamma^2},
\frac{\ln(\frac{\log(\nicefrac{\alpha_{p_{\max}}}{\alpha_{p_{\min}}})}{\beta})}{\epsilon}
\Bigg\}
\end{aligned}
\]
for a universal constant $C>0$, then with probability at least $1-\beta$ the outputs $(\tilde{\alpha}_p,\tilde{x}_m)$ satisfy
\[
\tilde{\alpha}_p \ \in\ \big[(1-\gamma)\alpha_p,\ (1+\gamma)\alpha_p\big]
\]
and
\[
\frac{\tilde{x}_m}{x_m}
\ \leq\ e^{\ \,c\,(\gamma/\alpha_p)\tau}
\]
for a universal constant $c>0$.

In particular, choosing $\tau = \frac{1}{4\ln(7)}$ yields a sufficient estimate of $x_m$ which yields $\mathrm{TV}\!\left(P(x_m,\alpha_p),\,P(\tilde x_m,\tilde\alpha_p)\right)\ \le\ \gamma$ by \Cref{cor:pareto-tv}.
\end{theorem}

\begin{proof}
    Privacy is immediate from \Cref{lem:pareto-both-privacy}. Utility follows from \Cref{lem:pareto-both-utility}.
\end{proof}

This application illustrates how exponential learning algorithms serve as a building block for differentially private estimation in more general parametric families.

\end{document}